\documentclass[12pt]{iopart}

\usepackage{graphicx,epsf,tabularx}
\usepackage{iopams,amsthm,bbm}
\usepackage[perpage]{footmisc}

\def\H{{\cal H}}
\def\B{{\cal B}}
\def\HC{{\cal HC}}
\def\HP{{\cal HP}}
\def\C{{\cal C}}
\def\D{{\cal D}}
\def\P{{\cal P}}
\def\PT{{\cal PT}}
\def\U{{\cal U}}
\def\Rset{\mathbb{R}}
\def\arccos{\mathrm{arccos\:}}
\def\id{\mathrm{id\:}}

\def\;{\, ; \, }
\def\Dset{\mathbb{D}}
\def\lp{\mathrm{LP}}

\renewcommand\vec[1]{\boldsymbol{#1}}

\newtheorem*{proposition}{Proposition}
\newtheorem{corollary}{Corollary}

\begin{document}

\title[General entropy-like uncertainty  relations in finite dimensions]{General
  entropy--like uncertainty relations in finite dimensions}

\author{S. Zozor$^{1,2}$, G. M. Bosyk$^{2,1}$, and M. Portesi$^{2,1}$}
\address{$^1$~Laboratoire  Grenoblois  d'Image,  Parole, Signal  et  Automatique
  (GIPSA-Lab, CNRS),  11 rue des Math\'ematiques, 38402  Saint Martin d'H\`eres,
  France}
\address{$^2$~Instituto   de   F\'{\i}sica  La   Plata   (IFLP),  CONICET,   and
  Departamento  de  F\'{\i}sica,   Facultad  de  Ciencias  Exactas,  Universidad
  Nacional de La Plata, C.C.~67, 1900 La Plata, Argentina}
\eads{\mailto{steeve.zozor@gipsa-lab.inpg.fr},
  \mailto{gbosyk@fisica.unlp.edu.ar}, \mailto{portesi@fisica.unlp.edu.ar}}


\begin{abstract}

  We revisit entropic formulations of the uncertainty principle for an arbitrary
  pair of positive operator-valued measures (POVM) $A$ and $B$, acting on finite
  dimensional   Hilbert  space.   Salicr\'u   generalized  $(h,\phi)$-entropies,
  including  R\'enyi and  Tsallis ones  among  others, are  used as  uncertainty
  measures associated  with the distribution probabilities  corresponding to the
  outcomes of the observables. We obtain a nontrivial lower bound for the sum of
  generalized entropies for any pair of entropic functionals, which is valid for
  both  pure  and  mixed states.   The  bound  depends  on the  overlap  triplet
  $(c_A,c_B,c_{A,B})$ with  $c_A$ (resp.  $c_B$)  being the overlap  between the
  elements of  the POVM $A$ (resp.   $B$) and $c_{A,B}$ the  overlap between the
  pair  of  POVM.   Our  approach  is   inspired  by  that  of  de  Vicente  and
  S\'anchez-Ruiz [Phys.\ Rev.\  A \textbf{77}, 042110 (2008)] and  consists in a
  minimization of the entropy sum  subject to the Landau--Pollak inequality that
  links the maximum probabilities of both observables.  We solve the constrained
  optimization problem in  a geometrical way and furthermore,  when dealing with
  R\'enyi  or Tsallis  entropic formulations  of the  uncertainty  principle, we
  overcome the H\"older conjugacy constraint  imposed on the entropic indices by
  the Riesz--Thorin theorem.  In the  case of nondegenerate observables, we show
  that for  given $c_{A,B} >  \frac{1}{\sqrt2}$, the bound obtained  is optimal;
  and  that,  for  R\'enyi  entropies,  our  bound  improves  Deutsch  one,  but
  Maassen--Uffink  bound  prevails  when  $c_{A,B}  \leq\frac12$.   Finally,  we
  illustrate by  comparing our bound  with known previous results  in particular
  cases of R\'enyi and Tsallis entropies.

\end{abstract}


\submitto{\JPA}

\maketitle


\section{Introduction}
\label{Introduction:sec}

The      uncertainty      principle~(UP),      originally     formulated      by
Heisenberg~\cite{Hei27}, is one the  most characteristic features of the quantum
world.  The  principle establishes  that one cannot  predict with  certainty and
simultaneously  the outcomes of  two (or  more) incompatible  measurements.  The
study  of quantitative  formulations of  this principle  has a  long outstanding
history.  First formulations  made use of variances as  uncertainty measures and
the principle was described state by state by the existence of a lower bound for
the  product  of  the   variances~\cite{Hei27,  Ken27,  Rob29}.   However,  such
formulations are not always adequate since the variance is not always convenient
for describing the uncertainty of  a random variable.  For instance, there exist
variables  with infinite  variance~\cite{SamTaq94}.   Moreover, in  the case  of
discrete-spectrum  observables, the  universal  (state-independent) lower  bound
becomes trivial  (zero), and thus  Heisenberg-like inequalities do  not quantify
the  UP~\cite{Deu83, MaaUff88, Lui01,  Lui11, Zoz12}.   For these  reasons, many
authors attempted  and still attempt to propose  alternative formulations, using
other    uncertainty   measures.     One   possibility    consists    in   using
information-theoretic   measures~\cite{Sha48,  Ren61,   CovTho06},   leading  to
entropic  uncertainty  relations~(EURs).   In  this line,  pioneering  works  by
Hirschman~\cite{Hir57}, Bialynicki-Birula and Mycielski~\cite{BiaMyc75} based on
important results due  to Beckner~\cite{Bec75}, Deutsch~\cite{Deu83}, or Maassen
and   Uffink   (MU)~\cite{MaaUff88}  who   proved   a   result  conjectured   by
Kraus~\cite{Kra87}, have  given rise to different formulations  of the principle
based  on Shannon  and generalized  one-parameter information  entropies,  or on
entropic  moments~\cite{Bia84, Raj95, San95,  PorPla96, San98,  GhiMar03, Bia06,
  ZozVig07, Lui07, VicSan08,  ZozPor08, WuYu09, WehWin10, BiaRud10, DehLop10:12,
  TomRen11,  ColYu11, BosPor11,  Ras11:03,Ras11:11,  Ras12, ColCol12,  BosPor12,
  ZozBos13, PucRud13,  FriGhe13, BosPor13, ColPia14,  RudPuc14}.  Versions using
the  sum  of  variances  (instead  of their  product)~\cite{Hua12},  the  Fisher
information~\cite{RomAng99,   RomSan06,  SanGon06},   or   moments  of   various
orders~\cite{ZozPor11} have also been developed.

In this  contribution, we  focus on  the formulation of  the UP  in the  case of
finite dimensions by using $(h,\phi)$-entropies, a generalization of the Shannon
entropy  due   to  Salicr\'u  \textit{et   al.}~\cite{SalMen93,  MenMor97}.   In
particular,  we deal  with two  well-known one-parameter  entropy  families, the
R\'enyi and Tsallis ones.  Our aim is to obtain a universal and nontrivial bound
for the sum of the entropies associated  with the outcomes of a pair of positive
operator-valued measures.   In order to do  this, we follow a  method similar to
that  of de  Vicente  and S\'anchez-Ruiz  in  Ref.~\cite{VicSan08}, solving  the
minimization  problem  for the  sum  of  generalized  entropies subject  to  the
Landau--Pollak inequality~\cite{LanPol61}.  We develop a geometrical approach to
the problem.

The paper  is organized as  follows. In Sec.~\ref{Statement:sec}, we  begin with
basic  definitions  and notation,  we  present  the  problem, and  we  summarize
previous  results on  EURs  that deal  with  R\'enyi or  Tsallis entropies.   In
Sec.~\ref{Main:sec}, we  give our  main results concerning  general entropy-like
formulations of  the~UP in finite dimensions.   For the sake  of comparison with
existing bounds in the  literature, in Sec.~\ref{Comparisons:sec} we choose some
particular cases.   A discussion is provided  in Sec.~\ref{Discussions:sec}. The
proofs of our results are given in detail in a series of appendices.


\section{Statement of the problem: notation and previous results}
\label{Statement:sec}


\subsection{Generalized entropies}

We  are interested in  quantitative formulations  of the  uncertainty principle,
particularly  through   the  use  of   information-theoretic  quantities.   More
precisely, as measure of ignorance or of lack of information we employ Salicr\'u
\textit{et al.}  $(h,\phi)$-entropies~\cite{SalMen93, MenMor97},
\begin{equation}
H_{(h,\phi)}(p) = h \left( \sum_{k=1}^{N} \phi(p_k) \right)
\label{SalicruEnt:eq}
\end{equation}
for any probability vector $p \in \P_N$ and where the {\em entropic functionals}
$\phi:  [0 \; 1]  \mapsto \Rset$  and $h:  \Rset \mapsto  \Rset$ are  such that,
either $\phi$ is concave  and $h$ is increasing, or $\phi$ is  convex and $h$ is
decreasing.
We restrict here to employ entropic functionals such that
\begin{itemize}
\item $\phi$ is continuous and strictly concave or strictly convex,
\item $h$ is continuous and strictly monotone,
\item  $\phi(0) = 0$  (so that  the ``elementary''  uncertainty associated  to a
  event with zero-probability is zero),
\item $h(\phi(1)) = 0$ (without loss of generality).
\end{itemize}
Many of  the well-known cases in  the literature satisfy  these assumptions (see
Refs.~\cite{SalMen93, MenMor97}  for a list  of examples). Among them,  the most
renowned ones are
\begin{itemize}
\item Shannon entropy~\cite{Sha48}, given by $\phi(x)  = - x \log x$ and $h(x) =
  x$ where $\log$ stands for the natural logarithm, corresponding to
  \begin{equation}
    H(p) = - \sum_k p_k \log p_k
    \label{Shannon:eq}
  \end{equation}
\item R\'enyi  entropies~\cite{Ren61}, introduced  in the domain  of mathematics
  from  the  same   axiomatics  as  Shannon  but  relaxing   only  one  property
  (recursivity is generalized); it is  given by $\phi(x) = x^\lambda$, and $h(x)
  = \frac{\log x}{1-\lambda}$, where $\lambda \geq 0$ is the entropic index,
  \begin{equation}
    R_\lambda(p) = \frac{1}{1-\lambda} \log \left( \sum_k p_k^{\,\lambda} \right)
    \label{Renyi:eq}
  \end{equation}
\item    Tsallis     entropies,    firstly    introduced     by    Havrda    and
  Charv\'at~\cite{HavCha67} from  an axiomatics quite close to  that of Shannon,
  then  by  Dar\'oczy~\cite{Dar70}  through  a generalization  of  a  functional
  equation satisfied by the Shannon entropy, and finally by Tsallis~\cite{Tsa88}
  in the domain  of nonextensive physics; it is given  by $\phi(x) = x^\lambda$,
  $\lambda \ge 0$, and $h(x) = \frac{x-1}{1-\lambda}$,
  \begin{equation}
    S_\lambda(p) = \frac{\displaystyle 1 -\sum_k p_k^{\,\lambda}}{\lambda-1}
    \label{Tsallis:eq}
  \end{equation}
\end{itemize}
The last two cases belong to  a general one-parameter family given by $\phi(x) =
x^\lambda$ and $h(x) = \frac{f(x)}{1-\lambda}$,
\begin{equation}
F_\lambda(p) = \frac{\displaystyle f \!\left( \sum_k \, p_k^{\, \lambda}
\right)}{1-\lambda}
\label{Flambda:eq}
\end{equation}
with $f$ increasing and $f(1) = 0$, and where the {\em entropic index} $\lambda$
plays the role of a ``magnifying  glass'', in the following sense: when $\lambda
<   1$,  the  contribution   of  the   different  terms   in  the   sum  $\sum_k
p_k^{\,\lambda}$ becomes more uniform with respect to the case $\lambda=1$, thus
stressing the  tails of  the distribution; conversely,  when $\lambda >  1$, the
leading probabilities of the distribution  are stressed in the summation.  As an
extreme example, for $\lambda = 0$  the generalized entropy $F_0(p)$ is simply a
function  of the number  of nonzero  components of  the probability  vector $p$,
regardless of the values of  these probabilities; this measure is closely linked
to the $l^0$ quasi-norm which measures the {\em sparsity} of a representation in
signal processing~\cite{RicTor13,  ElaBru02, GhoJam11}.  If  additionally $f$ is
differentiable,  with  $f'(1)  =  1$,  the Shannon  entropy  is  recovered  from
$F_\lambda$ entropies when $\lambda \to 1$.

The   generalized    $(h,\phi)$-entropies~\eref{SalicruEnt:eq}   satisfy   usual
properties as:
\begin{itemize}
\item $H_{(h,\phi)}(p)$ is a Schur-concave function of its argument, that is, if
  $p$ is majorized\footnote{By definition, $p \prec q$ means that, $\sum_{k=1}^m
    p^\downarrow_{\,k} \le \sum_{k=1}^m q^\downarrow_{\,k}, m = 1, \ldots, N-1$,
    and  $\sum_{k=1}^N p_k  = \sum_{k=1}^N  q_k$ where  $\cdot^\downarrow$ means
    that the components  are rearranged in decreasing order.}   by $q$, which is
  denoted  $p  \prec  q$,  then  $H_{(h,\phi)}(p)  \ge  H_{(h,\phi)}(q)$.   This
  property is a consequence of Karamata inequality that states that if $\phi$ is
  convex  (resp.\ concave), then  $p \mapsto  \sum_k \phi(p_k)$  is Schur-convex
  (resp.\                    Schur-concave)                    (see~\cite{Kar32}
  or~\cite[Chap.~3,~Prop.~C.1]{MarOlk11}), together  with the decreasing (resp.\
  increasing) property  of $h$.   The property of  Schur-concavity is  useful in
  some     problems    of     combinatorial,     numerical    or     statistical
  analysis~\cite{MarOlk11}.
\item $H_{(h,\phi)}(p)  \ge 0 \  \ \forall\ p  \in \P_N$, with equality  iff the
  probability  distribution  is a  Kronecker  delta:  $p_k  = \delta_{k,i}$  for
  certain $i$,  that is,  the $i$th-outcome appears  with certainty so  that the
  ignorance  is zero.   This property  is  a consequence  of Schur-concavity  of
  $H_{(h,\phi)}$ since  $p \prec [1 \quad  0 \quad \cdots  \quad 0]^t$, together
  with $h(\phi(1)) = 0$.
\item  $H_{(h,\phi)}(p) \le  h\!\left(  N \phi  \!   \left( \frac{1}{N}  \right)
  \right)  \  \  \forall\  p  \in  \P_N$,  with  equality  iff  the  probability
  distribution  is  uniform: $p_k  =  \frac{1}{N}$ for  all  $k$,  that is,  all
  outcomes appear  with equal  probability so that  the uncertainty  is maximal.
  Again,  this property is  a consequence  of Schur-concavity  of $H_{(h,\phi)}$
  since $\left[  \frac{1}{N} \quad \cdots  \quad \frac{1}{N} \right]^t  \prec p$
  (see~\cite[Eq.~(8),~p.~9]{MarOlk11}).
\item $H_{(h,\phi)}(p)$ is a concave function  of $p$ if $h$ is concave; this is
  due  to the  facts  that: (i)  for  concave (resp.\  convex) function  $\phi$,
  function    $p    \mapsto     \sum_k    \phi(p_k)$    is    concave    (resp.\
  convex)~\cite{CamMar09},   and  (ii)  function   $h$  is   increasing  (resp.\
  decreasing).  This property is useful in optimization problems~\cite{AndEvg07,
    CamMar09}. Shannon entropy is  known to be concave~\cite{CovTho06}.  R\'enyi
  entropy is concave  for $\lambda \in [0 \;  1]$; and in fact, it  can be shown
  that there exists an $N$-dependent  index $\lambda_*(N)$ greater than 1, up to
  which R\'enyi entropy remains concave~\cite[p.~57]{BenZyc06}.  Tsallis entropy
  is concave for any index $\lambda \ge 0$.
\end{itemize}

Furthermore, the  one-parameter entropy $F_\lambda$ is a  decreasing function in
terms of $\lambda$ for fixed $p$.   With the positivity of $f$, this ensures the
convergence of $F_\lambda$ (at least simply) when $\lambda \to + \infty$ so that
$F_\infty$ could  be called {\em minimal  generalized $F_\lambda$-entropy} (when
the limit is not identically zero).

Finally, note that from the strict  monotony of the function $h$, there exists a
one-to-one mapping between two generalized entropies sharing the same functional
$\phi$, say $(h,\phi)$ and $(g,\phi)$, under the form $H_{(h,\phi)}(p) = h \big(
g^{-1}  \big(  H_{(g,\phi)}(p)  \big)  \big)$.   For  instance,  the  one-to-one
mappings     between     R\'enyi     entropy~\eref{Renyi:eq}     and     Tsallis
entropy~\eref{Tsallis:eq}, for a given $\lambda$, are
\begin{equation}
S_\lambda(p) = \frac{1 - \exp\big( (1-\lambda) \, R_\lambda(p) \big)}{\lambda-1}
\label{RenyiToTsallisMapping:eq}
\end{equation}
and
\begin{equation}
R_\lambda(p) = \frac{1}{1-\lambda} \log \big( 1 + (1-\lambda) \, S_\lambda(p)
\big).
\label{TsallisToRenyiMapping:eq}
\end{equation}


\subsection{Entropic uncertainty relations}

Let $\H$ be an $N$-dimensional  Hilbert space.  A general quantum measurement is
described by positive operator-valued measures (POVM).  This is a set $A = \{A_i
\}_{i=1}^{N_A}$  of  Hermitian positive  semidefinite  operators satisfying  the
completeness  relation $\sum_{i=1}^{N_A}  A_i=  I$, where  $I$  is the  identity
operator and  $N_A$ is the number of  outcomes.  For given POVM  $A$ and quantum
system described by a  density operator $\rho$ (Hermitian, positive semidefinite
with unit trace)  acting on $\H$, the probability of the  $i$th outcome is equal
to $p_i(A,\rho) = \Tr (A_i \rho)$.

In this contribution,  we consider the $(h,\phi)$-entropies~\eref{SalicruEnt:eq}
for the probability vectors
\begin{eqnarray}
p(A,\rho) & = & \left[p_1(A,\rho) \: \cdots \: p_{N_A}(A,\rho) \right]^t \quad
\mbox{with} \quad p_i(A,\rho) = \Tr (A_i \rho) \quad \mbox{and} \nonumber \\
p(B,\rho) & = & \left[p_1(A,\rho) \: \cdots \: p_{N_B}(B,\rho) \right]^t \quad
\mbox{with} \quad p_j(B,\rho) = \Tr (B_j \rho), \nonumber
\end{eqnarray}
associated  with  the  measurements  of two POVM  $A$  and  $B$,  respectively.

The fact that the sum of $(h,\phi)$-entropies is lower bounded gives rise to
an entropy-like formulation of the~UP, that is, inequalities of the form
\begin{equation}
H_{(h_A,\phi_A)} \big( p(A,\rho) \big) + H_{(h_B,\phi_B)} \big( p(B,\rho) \big) \ge
\B_{(h_A,\phi_A),(h_B,\phi_B)}
\label{F-UP:eq}
\end{equation}
for  any two pairs  $(h_A,\phi_A)$ and  $(h_B,\phi_B)$ of  entropic functionals,
where the  bound $\B_{(h_A,\phi_A),(h_B,\phi_B)}$ is  nontrivial, i.e., nonzero,
and  universal in the  sense of  being independent  of the  state $\rho$  of the
quantum system.  In particular, dealing with the family $F_\lambda$, we focus on
the case where  $f$ is the same  for both entropies, but with  an arbitrary pair
$(\alpha,\beta)$ of nonnegative entropic indices.   The ultimate goal is to find
the  optimal bound,  which  by definition  is  obtained by  minimization of  the
left-hand side, i.e.,
\begin{equation}
\overline{\B}_{(h_A,\phi_A),(h_B,\phi_B)}(A,B) \equiv \min_\rho \: \left\{
H_{(h_A,\phi_A)} \big( p(A,\rho) \big) + H_{(h_B,\phi_B)} \big( p(B,\rho) \big)
\right\}
\label{BoundAsAMin:eq}
\end{equation}

In the case of two nondegenerate quantum measurements, the optimal bound depends
on the transformation matrix $T$ whose entries are given by
\begin{equation}
T_{ij} = \langle b_j | a_i \rangle,
\label{Relation_psi_psitilde:eq}
\end{equation}
where   $\left\{  |a_i\rangle  \right   \}_{i=1}^N$  and   $\left\{  |b_j\rangle
\right\}_{j=1}^N$  are eigenbases of  $A$ and  $B$, respectively  ($A_i =  | a_i
\rangle \langle a_i |$,  $B_j = | b_j \rangle \langle b_j |$,  $N_A = N_B = N$).
From the  orthonormality of the bases,  $T \in \U(N)$ where  $\U(N)$ denotes the
set  of $N  \times N$  unitary matrices.   A relevant  characteristic of  such a
unitary matrix is its greatest-modulus element,
\begin{equation}
c(T) = \max_{i,j} | \langle b_j | a_i \rangle |,
\label{Overlap:eq}
\end{equation}
the so-called  {\em overlap} between  the eigenbases of  $A$ and $B$.   From the
unitary  property of  matrix $T$,  the overlap  is in  the range  $c  \in \left[
  \frac{1}{\sqrt N} \; 1 \right]$.  The case $c = \frac{1}{\sqrt N}$ corresponds
to $A$  and $B$ being  \textit{complementary observables}, meaning  that maximum
certainty in  the measure of  one of them,  implies maximum ignorance  about the
other.  In the  opposite extreme case, $c=1$ corresponds  to observables $A$ and
$B$ sharing (at  least) an eigenvector; this situation  happens for example when
the observables commute.

In  this nondegenerate  context,  to find  the  optimal bound  depending on  the
transformation matrix is a difficult problem  in general; a weaker problem is to
restrict to bounds  depending on the overlap $c$ instead of  on the whole matrix
$T$. Thus, the optimal $c$-dependent bound writes
\begin{equation}
\widetilde{\B}_{(h_A,\phi_A),(h_B,\phi_B);N}(c) \:\: = \min_{T \in \U(N):\: c(T)
= c} \overline{\B}_{(h_A,\phi_A),(h_B,\phi_B)}(T)
\label{BoundOverlap:eq}
\end{equation}
We call  $\widetilde{\B}_{(h_A,\phi_A),(h_B,\phi_B);N}(c)$ the {\em $c$-optimal}
bound        in        order         to        distinguish        it        from
$\overline{\B}_{(h_A,\phi_A),(h_B,\phi_B)}(T)$  that we  call  {\em $T$-optimal}
bound.

Similarly,  in the  general POVM  framework, finding  the  $(A,B)$-optimal bound
Eq.~\eref{BoundAsAMin:eq}  is a  difficult  task. In  this  context, a  relevant
characteristic of the pair $(A,B)$ is the triplet of overlaps,
\begin{equation}
\vec{c}(A,B) = ( c_A, c_B , c_{A,B} ) \quad \mbox{where } \quad \left\{
\begin{array}{lll}
c_A & = & \displaystyle \max_i \| \sqrt{A_i} \| \\[3mm]
c_B & = & \displaystyle \max_j \| \sqrt{B_j} \| \\[3mm]
c_{A,B} & = & \displaystyle \max_{i,j} \| \sqrt{A_i} \, \sqrt{B_j} \|
\end{array}\right.
\label{OverlapSet:eq}
\end{equation}
[in the nondegenerate case, $\vec{c} =  (1,1,c)$].  A weaker problem is again to
restrict  to bounds  depending only  on $\vec{c}$,  the  $\vec{c}$-optimal bound
being
\begin{equation}
\widetilde{\B}_{(h_A,\phi_A),(h_B,\phi_B); \vec{N}}(\vec{c}) = \min_{(A,B):\:
\vec{c}(A,B) = \vec{c}} \overline{\B}_{(h_A,\phi_A),(h_B,\phi_B)}(A,B)
\label{BoundOverlapSet:eq}
\end{equation}
with $\vec{N} = (N_A,N_B,N)$.

\hfill

The study  of entropic formulations to quantify  the~UP is not new  and has been
addressed in various contexts~\cite{Deu83, Bia84, Kra87, MaaUff88, Raj95, San95,
  PorPla96, San98, GhiMar03, Bia06, Lui07, ZozVig07, ZozPor08, VicSan08, WuYu09,
  WehWin10,   BiaRud10,  DehLop10:12,  BosPor11,   Ras11:03,Ras11:11,  TomRen11,
  ColYu11,  BosPor12, ColCol12, Ras12,  ZozBos13, PucRud13,  FriGhe13, BosPor13,
  ColPia14, RudPuc14}.  However, the problem of finding $\vec{c}$-optimal (resp.
$c$-optimal) or $(A,B)$-optimal (resp.  $T$-optimal) bounds in the form posed in
Eqs.~\eref{F-UP:eq}--\eref{BoundOverlapSet:eq} still remains open in many cases.
Moreover, many available results correspond to R\'enyi or Tsallis entropies with
conjugated       indices       (in       the      sense       of       H\"older:
$\frac{1}{2\alpha}+\frac{1}{2\beta}=1$) as  they are based  on the Riesz--Thorin
theorem~\cite{HarLit52};  however,  recently   some  results  were  derived  for
nonconjugated indices in some particular situations.

For  the sake  of  later comparison  we  summarize existing  bounds, dealing  in
particular with R\'enyi or Tsallis entropies, classified by the entropic measure
used  and  the  entropic indices  involved.   To  fix  notation, we  define  the
following regions in the $\alpha$--$\beta$-plane:
\begin{equation}\left\{\begin{array}{l}
\C = \left\{ (\alpha,\beta) \in \left( \frac12 \; + \infty \right)^2: \:
\beta = \frac{\alpha}{2 \alpha-1} \right\} \vspace{2.5mm}\\
\underline{\C} = \left[ 0 \; \frac12 \right] \times \Rset_+ \ \ \bigcup \:
\left\{ (\alpha,\beta) \in \Rset_+^{\,2}: \: \alpha > \frac12, \: \beta <
\frac{\alpha}{2 \alpha-1} \right\} \vspace{2.5mm}\\
\overline{\C} = \left\{ (\alpha,\beta) \in \Rset_+^{\,2}: \: \alpha > \frac12, \:
\beta > \frac{\alpha}{2 \alpha-1} \right\}
\end{array}\right.\label{Regions:eq}\end{equation}
which  are  called conjugacy  curve  and  regions  ``below'' and  ``above''  the
conjugacy curve, respectively (see Fig.~\ref{Regions:fig}).
\begin{figure}[htbp]
\begin{tabular}
{
>{}m{.31\textwidth}
>{}m{.66\textwidth}
}
\centerline{\includegraphics[height=3.75cm]{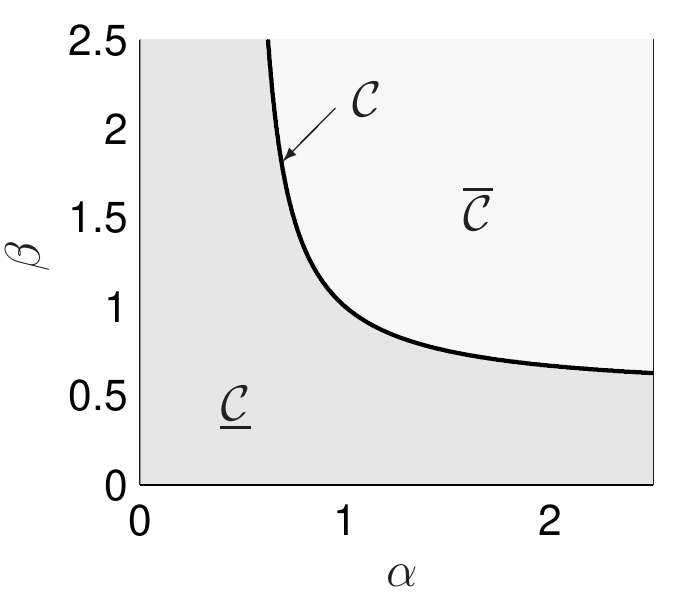}}
&
\vspace{-5mm}
\caption{The conjugacy curve $\C$ is represented by the solid line (the positive
  branch  of the  hyperbola  $\frac{1}{2\alpha}+\frac{1}{2\beta}=1$), while  the
  region $\underline{\C}$ ``below''  this curve is in dark  gray, and the region
  $\overline{\C}$ ``above'' that curve is represented in light gray.}
\label{Regions:fig}
\end{tabular}
\end{figure}\newline
Results available in the literature comprise the following:
\begin{itemize}
\item  Shannon entropy: $(\alpha,\beta)=(1,1)$
    \begin{itemize}
    \item[$\diamond$]  Deutsch obtained  the first  bound in  1983 \cite{Deu83},
      which is given by $\B^D(c) = - 2 \log \left( \frac{1+c}{2} \right)$.
    \item[$\diamond$]  MU  improved Deutsch  bound  by  using the  Riesz--Thorin
      theorem, in the context of pure  states.  Their bound is $\B^{MU}(c) = - 2
      \log c$  \ and  it is not  optimal, except for  complementary observables,
      that is, for $c = \frac{1}{\sqrt{N}}$.
    \item[$\diamond$]  de  Vicente  and  Sanchez-Ruiz~\cite{VicSan08,  BosPor11}
      improved MU bound  in the range $c \in [c^* \;  1]$ with $c^*\simeq 0.834$
      by using the Landau--Pollak inequality that links $\max_i p_i(A,\rho)$ and
      $\max_j p_j(B,\rho)$,  in the context of  pure states.  This  bound is not
      optimal,  except  for  complementary  observables  (see  also~\cite{Bia06,
        ZozPor08}) or for qubits ($N=2$)~\cite{GhiMar03, ZozBos13}.
    \item[$\diamond$]  Recently, Coles  and Piani  (CP)~\cite{ColPia14} improved
      the MU bound  in the whole range of the overlap  $c$, indeed they obtained
      the bound  $\B^{CP}(c,c_2) = -  2 \log c  + (1-c) \,  \log \frac{c}{c_2}$,
      where $c_2$ is  the second largest value among  the $|T_{ij}|$.  Moreover,
      the authors obtained a stronger but implicit bound $\B^{\overline{CP}}(T)$
      and  generalized their  results  for POVMs  and  bipartite scenarios  (see
      also~\cite{RudPuc14}).
    \end{itemize}
\item R\'enyi entropies:
  \begin{itemize}
  \item[$\diamond$]  For  $(\alpha,\beta) \in  \C$,  the  MU bound  $\B^{MU}(c)$
    remains valid.   Rastegin extended this result  to the case  of mixed states
    and generalized  quantum measurements~\cite{Ras10, Ras12}.   These works are
    mainly based on  Riesz--Thorin theorem.  The bound is  not tight, except for
    $c=\frac{1}{\sqrt{N}}$~\cite{Bia06, ZozPor08}.
  \item[$\diamond$]  For  $(\alpha,\beta)  \in  \underline{\C}$,  the  MU  bound
    $\B^{MU}(c)$ remains valid due to the decreasing property of R\'enyi entropy
    with  the index.   Here again,  for $c  = \frac{1}{\sqrt{N}}$  the  bound is
    optimal~\cite{Bia06, ZozPor08}.
  \item[$\diamond$]  For $(\alpha,\beta) \in  \overline{\C}$, the  Deutsch bound
    $\B^D(c)$  remains  valid.   This  result  is  due  to  MU  who  solved  the
    minimization of the  sum of min-entropies (infinite indices)  subject to the
    Landau--Pollak  inequality.  Note  that the  Deutsch bound  is valid  in the
    whole  positive quadrant  (but  it is  not  optimal) due  to the  decreasing
    property of the R\'enyi entropy vs the index.
  \item[$\diamond$]   For   $\beta   =   \alpha$,   Pucha{\l}a,   Rudnicki   and
    {\.Z}yczkowski (PRZ)  in Ref.~\cite{PucRud13}  derived recently a  series of
    $N-1$  bounds   depending  on  the   transformation  matrix  $T$   by  using
    majorization  technique.   We   denote  by  $\B_{\alpha;\log}^{PRZ}(T)$  the
    greatest  of those  bounds which  is  not $T$-optimal  although it  improves
    previous ones in several situations.   A particular bound of the series (the
    worst  one)   depends  only  on  the   overlap  $c$,  and   expresses  as  $
    \frac{1}{1-\alpha}  \log \left[  \left( \frac{1+c}{2}  \right)^{2  \alpha} +
      \left( 1  - \left( \frac{1+c}{2} \right)^2 \right)^\alpha  \right]$ but it
    is  not $c$-optimal. Further  extensions of  this work  to mixed  states and
    generalized     quantum    measurements     are    given     by    Friedland
    \etal~\cite{FriGhe13}.
  \item[$\diamond$] For $(\alpha,\beta) \in [0,1]^2$, the CP bounds remain valid
    due to the decreasing property of R\'enyi entropy with the index.
  \item[$\diamond$] For $(\alpha,\beta) \in \Rset_+^{\,2}$ and $N=2$, we derived
    recently  the $T$-optimal  bound  $\overline{\B}_{\alpha,\beta;\log}(T)$. It
    depends  only  on the  overlap,  so  that it  is  $c$-optimal  as well,  and
    $\overline{\B}_{\alpha,\beta;\log}(T)                                       =
    \widetilde{\B}_{\alpha,\beta;\log;2}(c)$~\cite{ZozBos13}.   Note  that  this
    equality is trivial since only  $c$ parametrizes all the $|T_{ij}|$ and that
    in this  case the  phases play  no role (due  to the  symmetry of  the Bloch
    sphere or from the  $Z-Y$ decomposition for a single qubit~\cite{ZozBos13}).
    Numerical solutions have been found in  the whole quadrant, and we have been
    able  to derive  analytical expressions  in some  regions. In  addition, the
    states that  correspond to the  bound were obtained,  in terms of  the whole
    matrix $T$.
\end{itemize}
\item Tsallis entropies:
  \begin{itemize}
  \item[$\diamond$]  For $\beta  = \alpha$  and pure  states, the  inequality $$
    S_\alpha \big( p(A,\rho) \big) + S_\alpha \big( p(B,\rho) \big) + (1-\alpha)
    S_\alpha  \big( p(A,\rho)  \big) S_\alpha  \big( p(B,\rho)  \big) \,  \ge \,
    \frac{1 - \left( \frac{1+c}{2} \right)^{2 (\alpha-1)}}{\alpha-1} $$ has been
    derived  in  Ref.~\cite{PorPla96}.   This   relation  can  be  viewed  as  a
    consequence of the fact that the sum of R\'enyi entropies with equal indices
    is    lower    bounded    by    the    Deutsch    bound,    together    with
    relation~\eref{RenyiToTsallisMapping:eq} linking  $S_\alpha$ and $R_\alpha$.
    This bound has been refined  to $\frac{1 - c^{2 (\alpha-1)}}{\alpha-1}$ when
    $\alpha \in \left[ \frac12 \; 1 \right]$, starting from the MU inequality in
    the  conjugacy curve,  and using  the decreasing  property of  $R_\alpha$ vs
    $\alpha$, and relation~\eref{RenyiToTsallisMapping:eq}.
  \item[$\diamond$]  For  $(\alpha,\beta) \in  \C$,  following  recent works  of
    Rastegin \cite{Ras11:03, Ras11:11}, one can obtain the inequality $$S_\alpha
    \big(  p(A,\rho) \big)  +  S_\beta \big(  p(B,\rho)  \big) \ge  \frac{1-c^{2
        (\lambda-1)}}{\lambda-1}   \equiv   \B_{\alpha,\beta;\id-1}^R(c)   \quad
    \mbox{with}  \quad \lambda  = \max\{\alpha,\beta\}$$  ($\id$ stands  for the
    identity function, $\id(x)=x$).
  \item[$\diamond$]    For    $(\alpha,\beta)    \in   \underline{\C}$,    bound
    $\B_{\alpha,\beta;\id-1}^R(c)$ remains valid  due to the decreasing property
    of Tsallis entropy vs the entropic index.
  \item[$\diamond$] For $(\alpha,\beta) \in  [0,1]^2$, MU, Deutsch and CP bounds
    remain  valid due to  the decreasing  property of  Tsallis entropy  with the
    index.
  \end{itemize}

  One can find  in the literature many bounds improving  the above mentioned, in
  special contexts  (particular overlap and/or  particular pair of  indices). We
  refer the interested reader  to \cite{San98, Lui07, Ras11:03, BosPor12, Ras12,
    BosPor13, BiaRud10, DehLop10:12, RudPuc14}.  For the sake of completeness of
  this  short review, it  is worth  mentioning that  there is  a new  insight of
  entropic uncertainty  relations that allows the  observer to have  access to a
  quantum   memory~\cite{BerChr10,   ColYu11,   TomRen11,  ColCol12,   MulDup13,
    FraLie13}.  Also, there exist entropic  formulations of the~UP for more than
  two  measurements (in  particular, for  mutually  unbiased bases)~\cite{Iva92,
    San95,  BalWeh07,  WuYu09, WehWin10}  and  for  observables with  continuous
  spectra~\cite{Bia06,  Hua11,  FraLie12,  BerChr13}.   These topics  have  many
  applications in different issues of quantum information such that entanglement
  detection,  proof of  the  security of  quantum  cryptographic protocols,  and
  others~\cite{Gio04,  GuhLew04,  Hua10,  NgBer12,  FurFra12,  SchBro14}.   Such
  studies go beyond the scope of the present paper.

  Finally, it can be shown that some bounds and relations discussed above can be
  expressed  in terms  of the  generalized entropies  of the  family $F_\lambda$
  (with  a common  function $f$  for both  entropies, but  any pair  of entropic
  indices):
\item $F_\lambda$ entropies:
  \begin{itemize}
  \item[$\diamond$] For  $(\alpha,\beta) \in  \C \cup \underline{\C}$,  with the
    additional  condition that  \ $x  f'(x)$ is  increasing, following  the same
    approach as that of Rastegin in Ref.~\cite{Ras11:03, Ras11:11} and using the
    decreasing  property  of  $F_\lambda$   vs  $\lambda$,  one  can  prove  the
    relation $$F_\alpha  \big( p(A,\rho) \big)  + F_\beta \big(  p(B,\rho) \big)
    \ge       \frac{f\left(c^{2      (\lambda-1)}\right)}{1-\lambda}      \equiv
    \B_{\alpha,\beta;f}^R(c)     \quad    \mbox{with}     \quad     \lambda    =
    \max\{\alpha,\beta\},$$ which includes as particular cases the results of MU
    and of Rastegin.
  \item[$\diamond$]   For  $\beta  =   \alpha  \ge   1$:  since   $F_\alpha$  is
    Schur-concave, the Corollary 2 of Ref.~\cite{PucRud13} allows us to derive a
    $T$-dependent bound  for $F_\alpha \big( p(A,\rho)  \otimes p(B,\rho) \big)$
    where  $\otimes$ denotes  the Kronecker  product\footnote{$[p_1  \:\: \cdots
      \:\: p_N]^t  \otimes [q_1 \:\:  \cdots \:\: q_M]^t  \: = \: [p_1  q_1 \:\:
      \cdots  \:\:  p_1 q_M  \:\:  \cdots  \:\: p_N  q_1  \:\:  \cdots \:\:  p_N
      q_M]^t$}. If $f(x) + f(y) \le f(xy)$ for $0 \le x, y \le 1$ then $F_\alpha
    \big( p(A,\rho)  \big) + F_\alpha  \big( p(B,\rho) \big) \ge  F_\alpha \big(
    p(A,\rho)  \otimes p(B,\rho)  \big)$. Applying  the  results of  PRZ to  the
    right-hand  side we obtain  a bound  for the  sum of  $F_\lambda$ entropies.
    R\'enyi and Tsallis  entropies with entropic index greater  than or equal to
    one are particular cases.
  \item[$\diamond$]  For $\beta =  \alpha \le  1$: from  the Schur  concavity of
    $F_\lambda$  we  have  again  a  $T$-dependent  bound  for  $F_\alpha  \big(
    p(A,\rho) \otimes p(B,\rho) \big)$. Now, if  $f(x) + f(y) \ge f(xy)$ for $x,
    y \ge 1$, one has $F_\alpha \big( p(A,\rho) \big) + F_\alpha \big( p(B,\rho)
    \big)  \ge F_\alpha \big(  p(A,\rho) \otimes  p(B,\rho) \big)$  (notice that
    Tsallis entropy does not fulfill this property in this case). Therefore, PRZ
    results applied  to the right-hand side  allows again to obtain  a bound for
    the sum of  this class of entropies.  R\'enyi  entropies with entropic index
    lower than or equal to one are particular cases.
  \item[$\diamond$] For $(\alpha,\beta) \in  [0,1]^2$, MU, Deutsch and CP bounds
    remain valid due to the  decreasing property of the entropy $F_\lambda$ with
    the index.
\end{itemize}
\end{itemize}


\section{Generalized entropic uncertainty relations}
\label{Main:sec}

We  extend results  summarized  in the  preceding  section for  POVM pairs,  and
generalized  entropies~\eref{SalicruEnt:eq}  with  arbitrary pairs  of  entropic
functionals $(h_A,\phi_A)$ and $(h_B,\phi_B)$.   Our approach follows that of de
Vicente  and  S\'anchez-Ruiz~\cite{VicSan08} except  that  here the  concomitant
optimization problem is  mainly solved in a geometrical way.   This allows us to
generalize the results to arbitrary  entropic functionals.  Moreover, we use the
fact that the Landau--Pollak inequality applies for POVM pairs and for both pure
and  mixed states~\cite{BosOsa14, BosZoz14}  to argue  that our  results include
these situations.

Our   major   results   are    given   by   the   following   Proposition,   and
Corollaries~\ref{CorollaryQubit:cor},~\ref{ImproveDeutsch:cor}
and~\ref{NoImprovedMU:cor}:
\begin{proposition}\label{CEUR:prop}
  Let us  consider a pair of  POVM $A = \{  A_i \}_{i=1}^{N_A}$ and $B  = \{ B_j
  \}_{j=1}^{N_B}$ acting on an  $N$-dimensional Hilbert space $\H$, and consider
  a quantum system described by a  density operator $\rho$ acting on $\H$.  Then
  for generalized entropies of the form~\eref{SalicruEnt:eq}, with any two pairs
  of  entropic  functionals  $(h_A,\phi_A)$  and $(h_B,\phi_B)$,  the  following
  uncertainty relation holds:
\begin{equation}
H_{(h_A,\phi_A)} \big( p(A,\rho) \big) + H_{(h_B,\phi_B)} \big( p(B,\rho) \big)
\ge \B_{(h_A,\phi_A),(h_B,\phi_B)}(\vec{c}(A,B))
\label{FEURs:eq}
\end{equation}
where  the overlap triplet $\vec{c}(A,B)  = \left(  c_A,c_B,c_{A,B}  \right)$ is
given by Eq.~\eref{OverlapSet:eq}, and the lower bound expresses as
\begin{equation}
\hspace{-2.25cm} \B_{(h_A,\phi_A),(h_B,\phi_B)}(\vec{c}) = \left\{\begin{array}{l}
  \displaystyle \D_{(h_A,\phi_A)}(\gamma_A) + \D_{(h_B,\phi_B)}(\gamma_B)
  \qquad \mbox{if} \quad
  \gamma_{A,B} \le \gamma_A + \gamma_B\\[5mm]
  \displaystyle \!\!\min_{\theta \in [\gamma_A, \gamma_{A,B} - \gamma_B]} \left(
  \D_{(h_A,\phi_A)}(\theta) + \D_{(h_B,\phi_B)}(\gamma_{A,B} - \theta) \right)
  \:\: \mbox{otherwise}
  \end{array}\right.
\label{Cab:eq}
\end{equation}
with
\begin{equation}
\gamma_A  \equiv  \arccos  c_A, \quad
\gamma_B  \equiv  \arccos  c_B, \quad
\gamma_{A,B}  \equiv  \arccos  c_{A,B}
\end{equation}
and
\begin{equation}
\D_{(h,\phi)}(\theta) \equiv h \! \left( \left\lfloor
\frac{1}{\cos^2\theta}\right\rfloor \phi \! \left( \cos^2 \theta \right) + \phi
\! \left( 1 - \left\lfloor \frac{1}{\cos^2\theta}\right\rfloor \cos^2\theta
\right) \right)
\end{equation}
where $\lfloor \cdot \rfloor$ indicates the floor part.

\end{proposition}
\begin{proof}
See~\ref{Proposition:app}.
\end{proof}

For the  sake of simplicity, when  dealing with $F_\lambda$  entropies (with the
same function $f$ for both observables), the bound is simply denoted
\begin{equation}
 \B_{\alpha,\beta;f} \equiv \B_{\left( \frac{f}{1-\alpha} , \id^\alpha \right) ,
\left( \frac{f}{1-\beta} , \id^\beta \right)}
\end{equation}

Let us note the following facts:
\begin{itemize}
\item  $\B_{(h_A,\phi_A),(h_B,\phi_B)}(\vec{c})$  is  explicitly independent  of
  $\vec{N} = (N_A,N_B,N)$.
\item Previous results  in the literature, in particular that  of de Vicente and
  S\'anchez-Ruiz \cite{VicSan08}, are extended here from Shannon to more general
  $(h,\phi)$-entropies,  the  former  being  recovered  as  a  particular  case.
  Moreover, our result applies in the POVM framework and for both pure and mixed
  states.
\item  For Tsallis entropies  with $\beta  = \alpha$,  it is  straightforward to
  obtain relations of the type
$$
S_\alpha \big( p(A,\rho) \big) + S_\alpha \big( p(B,\rho) \big) + (1-\alpha)
S_\alpha \big( p(A,\rho) \big) S_\alpha \big( p(B,\rho) \big) \ge \frac{1 -
e^{(1-\alpha) \, \B_{\alpha,\alpha;\log}(\vec{c})}}{\alpha-1}
$$
that improve and generalize the findings in~\cite{PorPla96} and is valid for all
positive entropic index.
\end{itemize}

Note   that,    except   when   $\gamma_{A,B}   \le    \gamma_A   +   \gamma_B$,
bound~\eref{Cab:eq} is implicit. This is also the case for several bounds in the
literature~\cite{VicSan08,  PucRud13,  ColPia14}.   But, as  for~\cite{VicSan08,
  ColPia14},  the problem  is shown  to  be reduced  to an  optimization on  one
parameter over  a bounded  interval, instead of  on $N(N-2)$  parameters. Notice
that  from  the  increasing  property  of  $\D_{(h,\phi)}(\theta)$  vs  $\theta$
(see~\ref{Proposition:app}), an explicit lower bound can be obtained:
\begin{corollary}
Whatever the overlaps triplet be, bound~\eref{Cab:eq} satisfies
\begin{equation}
\B_{(h_A,\phi_A),(h_B,\phi_B)}(\vec{c})   \ge  \D_{(h_A,\phi_A)}(\gamma_A)  +
  \D_{(h_B,\phi_B)}(\gamma_B).
\end{equation}
Thus the  expression on the  right hand side  lower bounds the entropy  sum even
when $\gamma_{A,B} > \gamma_A + \gamma_B$.
\end{corollary}
\noindent  Note  however that  this  analytic bound  is  weaker,  and that  when
$\gamma_A = \gamma_B = 0$ it turns out to be trivial.

Finally,  it  is  to be  noticed  that  bound~\eref{Cab:eq}  is in  general  not
$\vec{c}$-optimal.   Indeed, our  method  for solving  the minimization  problem
first treats separately  the contribution of each observable  in the entropy sum
and,  only in  a second  step the  link between  the observables  is  taken into
account  through the Landau--Pollak  inequality.  In  some specific  cases, this
relative weakness disappears, as we see now.

Hereafter, we consider the case  of nondegenerate quantum observables.  In this
case, we have $N_A =  N_B = N$, $c_A = c_B = 1$ ($\gamma_A  = \gamma_B = 0$) and
$c_{A,B}  =   c$  ($\gamma_{A,B}  >   0$  except  when   $c  =  1$),   then  the
bound~\eref{Cab:eq} reduces to
\begin{equation}
\B_{(h_A,\phi_A),(h_B,\phi_B)}(c) = \min_{\theta \in [0, \gamma]}
\left( \D_{(h_A,\phi_A)}(\theta) + \D_{(h_B,\phi_B)}(\gamma - \theta) \right)
\label{Cabnondeg:eq}
\end{equation}
with $\gamma = \arccos c$.

As already  mentioned, bound~\eref{Cabnondeg:eq} is in  general not $c$-optimal.
However, it can  be shown that this bound  does turn out to be  optimal for some
particular  values  of  the  overlap.   This  is  summarized  in  the  following
corollary:
\begin{corollary}
\label{CorollaryQubit:cor}
When   $c    >   \frac{1}{\sqrt2}$    and   $N=2$   or    $N   \geq    4$,   the
bound~\eref{Cabnondeg:eq} is $c$-optimal,
\begin{equation}
\hspace{-22mm}\widetilde{\B}_{(h_A,\phi_A),(h_B,\phi_B);N}(c) =
\widetilde{\B}_{(h_A,\phi_A),(h_B,\phi_B);2}(c) = \min_{\theta \in [0 \;
\gamma]} \left( \D_{(h_A,\phi_A)}(\theta) + \D_{(h_B,\phi_B)}(\gamma-\theta)
\right)
\end{equation}
\end{corollary}
\begin{proof}
See~\ref{CorollaryQubit:app}.
\end{proof}
We suspect that  this corollary is also  valid when $N=3$, but we  have not been
able to prove it yet.

A  consequence of  the corollary  is that,  in  the range  of the  overlap $c  >
\frac{1}{\sqrt 2}$, the bound~\eref{Cabnondeg:eq}  reduces to the qubit case and
improves  all  $c$-dependent bounds  such  as  those  of MU  \cite{MaaUff88}  or
Rastegin  \cite{Ras11:03,   Ras11:11}  in  the  context  of   entropies  of  the
$F_\lambda$     family.      In      particular,     since     $\B^{MU}$     and
$\B_{\alpha,\beta;f}^R(c)$ do not depend on $N$, then $\B_{\alpha,\beta;\log}(c)
\ge \B^{MU}$ \ and \ $\B_{\alpha,\beta;f}(c) \ge \B_{\alpha,\beta;f}^R(c)$ \ for
any  $c  \ge  \frac{1}{\sqrt2}$ and  any  $N  \ge  2$.   Moreover, it  is  shown
in~\cite{ZozBos13}  that, for a  certain range  of entropic  indices and  in the
context  of  R\'enyi  entropies,  this  $c$-optimal bound  takes  an  analytical
expression.

Now, we  particularize the Proposition to  the case of  R\'enyi entropy [setting
$\phi(x) = x^\lambda$ and $f(x) = \frac{\log x}{1-\lambda}$, i.e., $f = \log$ in
the $F_\lambda$  family], which is  mostly used in  the literature of  EURs, and
compare  our  bound with  previous  ones,  as we  detail  in  the following  two
corollaries:
\begin{corollary}\label{ImproveDeutsch:cor}
  In  the context of  R\'enyi entropy,  the bound~\eref{Cabnondeg:eq}  is higher
  than that of Deutsch:
\begin{equation}
\B_{\alpha,\beta;\log}(c) \ge \B^D(c) = - 2 \log \left( \frac{1+c}{2} \right)
\end{equation}
\end{corollary}
\begin{proof}
See~\ref{CorollaryDeutsch:app}.
\end{proof}
This   result   is  particularly   interesting   above   the  conjugacy   curve,
$(\alpha,\beta)  \in  \overline{\C}$, where  the  only $c$-dependent  explicitly
known bound for R\'enyi entropies is precisely $\B^D(c)$.

It  is known  that  the sum  of  R\'enyi entropies  below  the conjugacy  curve,
$(\alpha,\beta) \in  \underline{\C}$, is lower bounded  by MU result.   For $c >
\frac{1}{\sqrt2}$  we   were  able  to   improve  this  bound,  but   for  $c\le
\frac{1}{\sqrt2}$ it is not always the case.  Indeed, we have:
\begin{corollary}\label{NoImprovedMU:cor}
  In the context  of R\'enyi entropy, when $c  \leq \frac12$ and $(\alpha,\beta)
  \in \underline{\C}$, the bound~\eref{Cabnondeg:eq} is lower than that of MU:
\begin{equation}
\B_{\alpha,\beta;\log}(c) \le \B^{MU}(c) = - 2 \log c
\end{equation}
\end{corollary}
\begin{proof}
See~\ref{CorollaryMU:app}.
\end{proof}
To the best of our knowledge, in  the range of the overlap $c \leq \frac{1}{2}$,
the  MU result  is the  tightest  $c$-dependent bound  when $(\alpha,\beta)  \in
\underline{\C}$.


\section{Comparison with previously known bounds}
\label{Comparisons:sec}


\subsection{Maassen--Uffink,   Rastegin and Coles--Piani   bounds}

We  now  compare our  bound  with previously  known  ones  in the  nondegenerate
context, for R\'enyi and Tsallis  entropies with indices $(\alpha,\beta)$ in the
region  $\C   \cup\underline{\C}$  or  just  within  $[0   \;  1]^2$.   Relative
differences       are      shown       through       density      plots       in
Figs.~\ref{CotaRenyi:fig},~\ref{CotaTsallis:fig},~\ref{CotaColes_Nw:fig}
and~\ref{CotaColes_Ns:fig},  for  chosen  typical  values of  the  overlap  $c$.
Positivity of these differences indicates that our bound improves the previous.

In   Fig.~\ref{CotaRenyi:fig}   we   plot   $\frac{\B_{\alpha,\beta;\log}(c)   -
  \B^{MU}(c)}{\B_{\alpha,\beta;\log}(c)}$ for entropic  indices in and below the
conjugacy curve,  $(\alpha,\beta) \in \C  \cup \underline{\C}$.  We  observe the
following behavior of our bound with respect to MU result:
\begin{itemize}
\item Up  to $c  = \frac12$  ($c = 0.5$  is shown),  the relative  difference is
  negative   or   zero,   so  our   bound   does   not   improve  the   MU   one
  (Corollary~\ref{NoImprovedMU:cor}).
\item  When $c$  is between  $\frac12$ and  $\frac{1}{\sqrt2}$ ($c  =  0.706$ is
  shown), the relative difference is positive or negative (although very small),
  so   our   bound   improves   the    MU   one   in   some   regions   of   the
  $\alpha$--$\beta$-plane.   This region  is delimited  by the  white  line: the
  improvement  takes place  below  this curve;  we  observe that  the region  of
  improvement increases with the overlap.
\item When $c$ exceeds $\frac{1}{\sqrt2}$ ($c = 0.708$ and $0.9$ are shown), the
  relative   difference   is   positive,   so   our  bound   improves   MU   one
  (Corollary~\ref{CorollaryQubit:cor}); the  improvement significantly increases
  with the overlap.
\end{itemize}

\begin{figure}[htbp]
\centerline{\includegraphics[width=\textwidth]{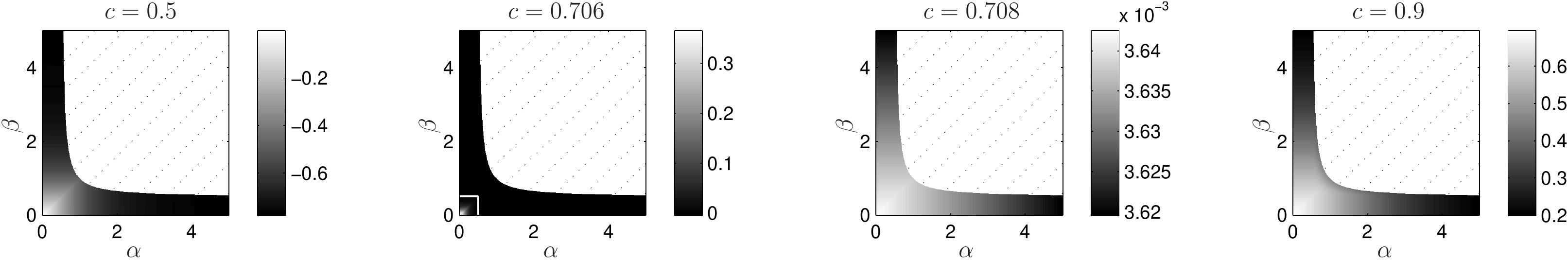}}
\caption{R\'enyi      entropy      case:       density      plots      of      \
  $\frac{\B_{\alpha,\beta;\log}(c)   -  \B^{MU}(c)}{\B_{\alpha,\beta;\log}(c)}$,
  for $(\alpha,\beta) \in  \C \cup \underline{\C}$ when $c  = 0.5, 0.706, 0.708$
  and $0.9$.}
\label{CotaRenyi:fig}
\end{figure}

In    Fig.~\ref{CotaTsallis:fig}    we    plot    the    relative    difference:
$\frac{\B_{\alpha,\beta;\id-1}(c)                                               -
  \B^R_{\alpha,\beta;\id-1}(c)}{\B_{\alpha,\beta;\id-1}(c)}$     for    entropic
indices  in  and  below  the   conjugacy  curve,  $(\alpha,\beta)  \in  \C  \cup
\underline{\C}$.   We observe the  following behavior  with respect  to Rastegin
results:
\begin{itemize}
\item Up to $c = \frac{1}{\sqrt2}$ ($c = 0.5$ and $0.6$ are shown), the relative
  difference is positive or negative, so  our bound improves the Rastegin one in
  some regions of the  $\alpha$--$\beta$-plane. The regions where an improvement
  occurs are outside  the domain marked by the black  line. These regions always
  exists (even when $c < \frac12$) and increases with the overlap.
\item When $c$ exceeds $\frac{1}{\sqrt2}$ ($c = 0.708$ and $0.9$ are shown), the
  relative  difference   is  positive,  so  our  bound   improves  Rastegin  one
  (Corollary~\ref{CorollaryQubit:cor})    and    the    improvement    increases
  significantly with the overlap.
\end{itemize}

\begin{figure}[htbp]
\centerline{\includegraphics[width=\textwidth]{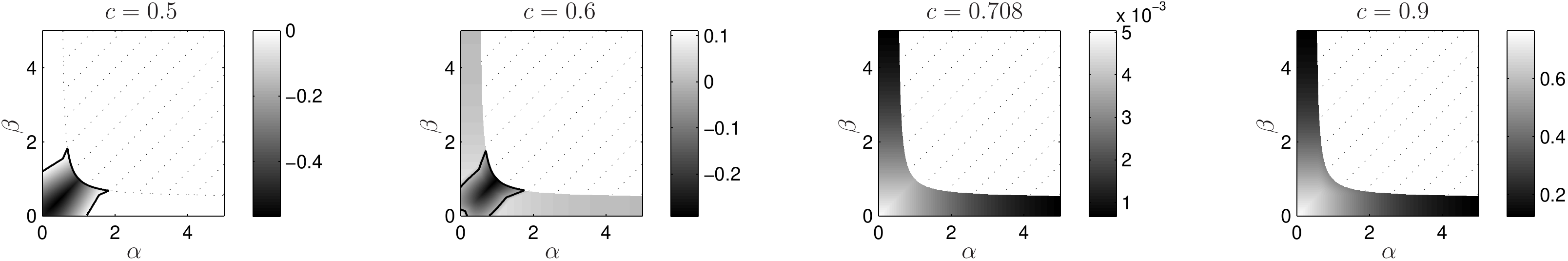}}
\caption{Tsallis      entropy      case:       density      plots      of      \
  $\frac{\B_{\alpha,\beta;\id-1}(c)              -             \B_{\alpha,\beta;
      \id-1}^R(c)}{\B_{\alpha,\beta;\id-1}(c)}$, for  $(\alpha,\beta) \in \C \cup
  \underline{\C}$ when $c = 0.5, 0.6, 0.708$ and $0.9$.}
\label{CotaTsallis:fig}
\end{figure}

In Figs.~\ref{CotaColes_Nw:fig} and~\ref{CotaColes_Ns:fig}  we plot the relative
differences:                   $\frac{\B_{\alpha,\beta;f}(c)                   -
  \B^{CP^\star}(c)}{\B_{\alpha,\beta;f}(c)}$,   for  $f   =  \id$   and  $\log$,
respectively,   where    $\B^{CP^\star}(c)   =\B^{CP}(c,c_2)$   with    $c_2   =
\frac{\sqrt{N-2+c^2}}{N-1}$ being the lowest possible second larger value of the
$|T_{ij}|$ (we  choose here  $N = 3$  and $N  = 10$ respectively);  the entropic
indices are $(\alpha,\beta) \in [0  \; 1]^2$.  We observe the following behavior
with respect to Coles--Piani results:
\begin{itemize}
\item For any value of $c$, the relative difference can be positive or negative,
  so  our  bound   improves  the  Coles--Piani  one  in   some  regions  of  the
  $\alpha$--$\beta$-plane. The regions where an improvement occurs are below the
  domain   marked   by    the   solid   line   in   Figs.~\ref{CotaColes_Nw:fig}
  and~\ref{CotaColes_Ns:fig}.   These regions  generally exist  (even when  $c <
  \frac12$) and  their extension  is greater with  the overlap  (the improvement
  always exists for $c \ge \frac{1}{\sqrt{2}}$).
\item  When  $N$  increases  (and  $c  <  \frac{1}{\sqrt{2}}$),  the  domain  of
  improvement  is smaller.   Remind  however  that the  best  possible CP  bound
  $\B^{CP}$ is plotted here.
\end{itemize}

\begin{figure}[htbp]
\centerline{\includegraphics[width=\textwidth]{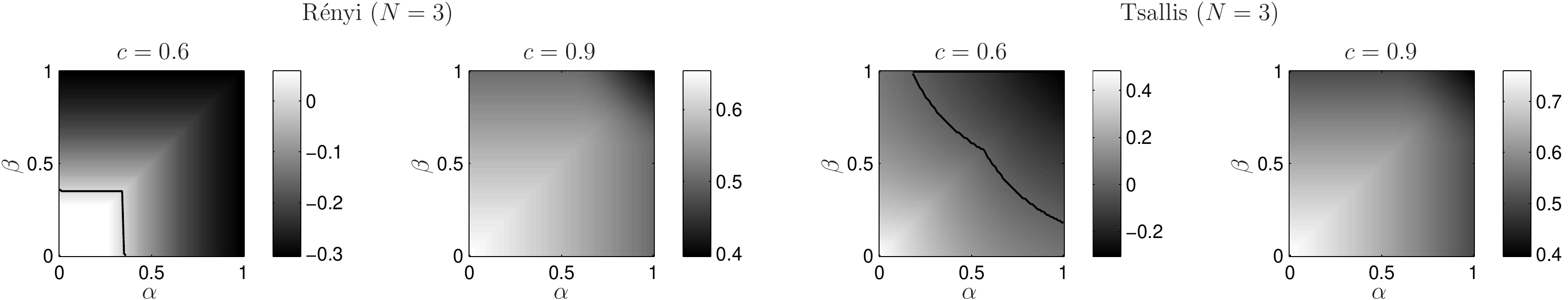}}
\caption{R\'enyi  and Tsallis  entropy cases  for $N  = 3$:  density plots  of \
  $\frac{\B_{\alpha,\beta;f}(c)  -  \B^{CP^*}(c)}{\B_{\alpha,\beta;f}(c)}$,  for
  $(\alpha,\beta) \in [0 \; 1]^2$ when $c = 0.6$ and $0.9$.}
\label{CotaColes_Nw:fig}
\end{figure}

\begin{figure}[htbp]
\centerline{\includegraphics[width=\textwidth]{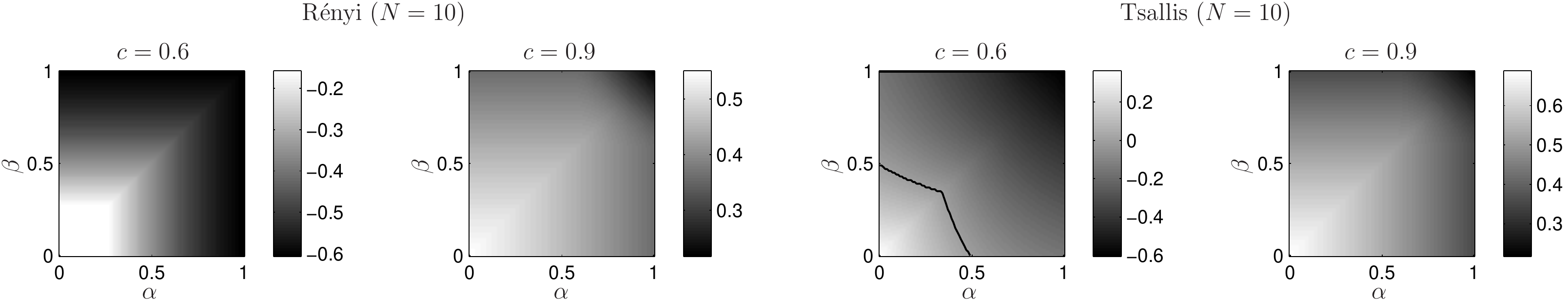}}
\caption{Same as Fig.~\ref{CotaColes_Nw:fig} for $N = 10$.}
\label{CotaColes_Ns:fig}
\end{figure}


\subsection{Bounds  for powers  of a  circular  permutation matrix  in the  line
  $\beta = \alpha$}

An illustrative  example to consider for  the evaluation of  generalized EURs is
given in Ref.~\cite{PucRud13}, where  a special class of transformation matrices
is used. Indeed,  the quantum observables here are  such that the transformation
between their  eigenbases is a  power of a circular  $N$-dimensional permutation
matrix, namely  $T_N(s) = \left[  \begin{array}{cc} 0 &  I_{N-1}\\ 1 &  0 \cdots
    0 \end{array}  \right]^s$ \  with $s  \in \left[ 0  \; \frac12  \right]$ and
where $I_{N-1}$  denotes the $(N-1)  \times (N-1)$ identity matrix.   We compute
our  bound in  these cases  for $N  =  3$ and  for some  chosen, equal  entropic
indices, and we  compare our results with  the bounds of PRZ, MU  and Deutsch in
the case of R\'enyi entropy (Fig.~\ref{PowerPermutationRenyi:fig}), and with the
bounds   of   Rastegin,  CP   and   PRZ  in   the   case   of  Tsallis   entropy
(Fig.~\ref{PowerPermutationTsallis:fig}).  In  this particular example, $T_N(s)$
can be analytically determined, allowing  for an analytic expression for both CP
bounds $\B^{CP}$ and $\B^{\overline{CP}}$.   It appears that, whatever $N$, both
bounds coincide and that they coincide with the MU bound.

In Fig.~\ref{PowerPermutationRenyi:fig}  we plot the  bounds $\B_{\alpha,\alpha;
  \log}(c)$,  $\B_{\alpha;\log}^{PRZ}(T)$, $\B^{MU}(c)$  and  $\B^D(c)$ for  the
R\'enyi  entropic formulation  of  the~UP, in  terms  of the  power  $s$ in  the
transformation  matrix,  when  $\alpha=0.8$  and $1.4$.   The  overlap  $c=c(s)$
corresponding to  the transformation $T=T(s)$ is  also shown in  the figure.  We
observe that:
\begin{itemize}
\item For  $\alpha = 0.8$  our bound improves  both PRZ and  MU ones for  a wide
  range of values of $s$. The fact that  our bound can be lower than that of PRZ
  for      $c       >      \frac{1}{\sqrt2}$      does       not      contradict
  Corollary~\ref{CorollaryQubit:cor}.   Indeed, the  PRZ bound  is $T$-dependent
  and is evaluated here for a particular $T$; it is not the minimum over all $T$
  for a given $c$.
\item     For    $\alpha=1.4$    our     bound    improves     Deutsch    result
  (Corollary~\ref{ImproveDeutsch:cor}) as well as PRZ for all $s$.
\end{itemize}

\begin{figure}[htbp]
\begin{tabular}
{
>{}m{.35\textwidth}
>{}m{.61\textwidth}
}
\includegraphics[height=3.75cm]{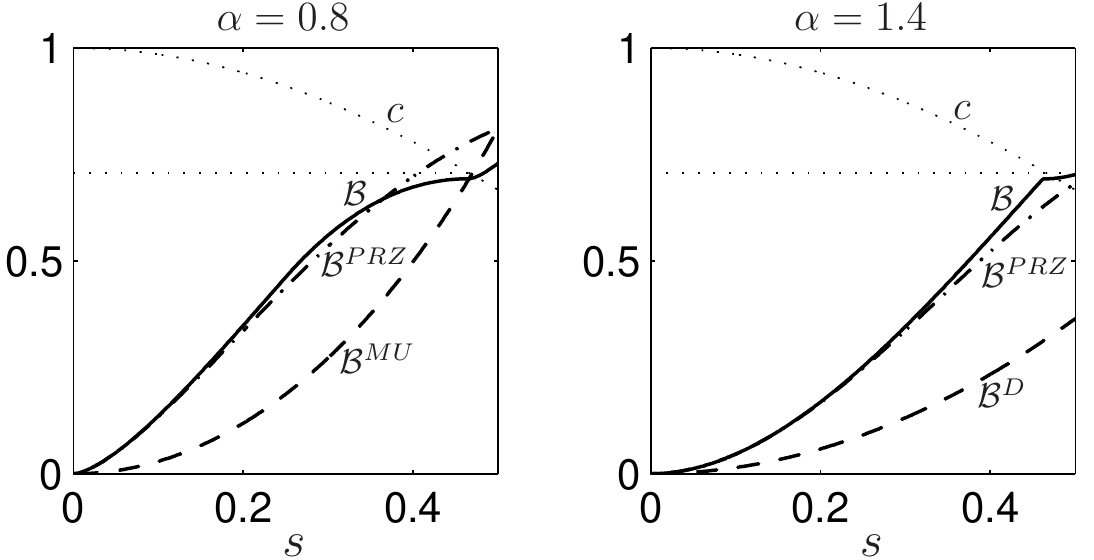}
&
\caption{R\'enyi  entropy case:  bounds  $\B \equiv  \B_{\alpha,\alpha;\log}(c)$
  (solid  line),   $\B^{PRZ}  \equiv  \B_{\alpha;\log}^{PRZ}(T)$  (dashed-dotted
  line), $\B^{MU} \equiv  \B^{MU}(c)$ (left plot, dashed line)  and $\B^D \equiv
  \B^D(c)$  (right  plot,  dashed line),  in  terms  of  the  power $s$  in  the
  transformation matrix for $\alpha = 0.8$  and $1.4$.  In addition, we plot the
  overlap $c$ in terms of $s$ (dotted line). }
\label{PowerPermutationRenyi:fig}
\end{tabular}
\end{figure}

In     Fig.~\ref{PowerPermutationTsallis:fig}     we     plot     the     bounds
$\B_{\alpha,\alpha;\id-1}(c)$,                   $\B_{\alpha;\alpha;\id-1}^R(c)$,
$\B^{\overline{CP}} = \B^{CP}  = \B^{MU}$, and $\B_{\alpha;\id-1}^{PRZ}(T)$, for
the Tsallis  entropic formulation of  the~UP, in terms  of the power $s$  in the
transformation matrix, when $\alpha=0.8$ and $1.4$. We observe that:
\begin{itemize}
\item For $\alpha=0.8$ our bound improves both Coles--Piani and Rastegin ones in
  a wide range of values of $s$.
\item For $\alpha=1.4$ our bound improves PRZ one for all $s$.
\end{itemize}


\begin{figure}[htbp]
\begin{tabular}
{
>{}m{.35\textwidth}
>{}m{.61\textwidth}
}
\includegraphics[height=3.75cm]{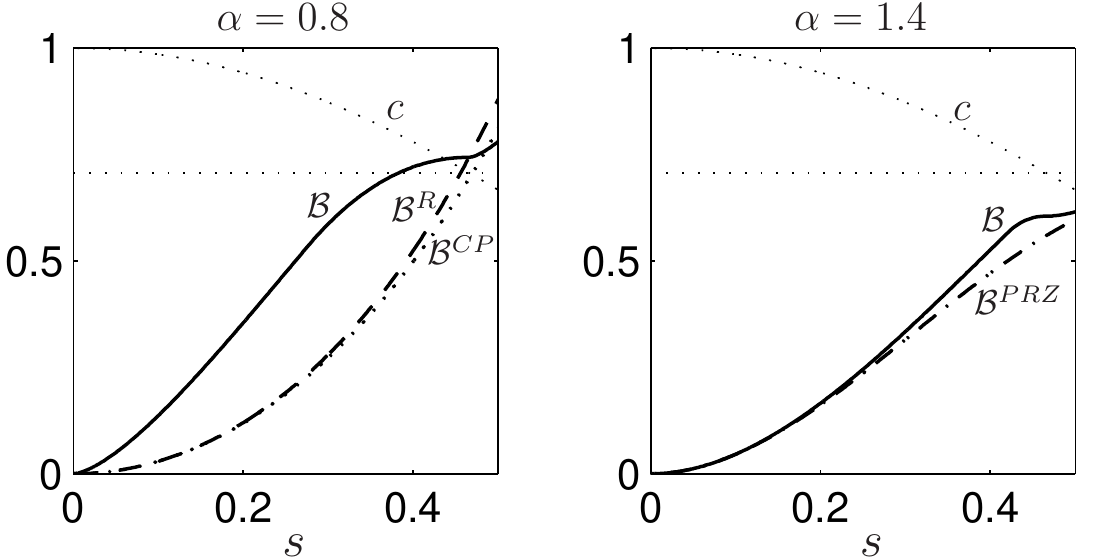}
&
\caption{Tsallis  entropy   case:  bounds  $\B_{\alpha,\alpha;\id-1}(c)$  (solid
  line), $\B_{\alpha,\alpha;\id-1}^R(c)$ (left  plot, dashed line), $\B^{CP}(T)$
  (left  plot,  dotted  line  below  that  of  $\B_{\alpha,\alpha;\id-1}^R(c)$),
  $\B_{\alpha;\id-1}^{PRZ}(T)$ (right plot, dashed-dotted line), in terms of the
  power  $s$  in the  transformation  matrix  for  $\alpha=0.8$ and  $1.4$.   In
  addition, we plot the overlap $c$ in terms of $s$ (dotted line).}
\label{PowerPermutationTsallis:fig}
\end{tabular}
\end{figure}


\subsection{Bounds   for   randomly  drawn   unitary   matrices   in  the   line
  $\beta=\alpha$}

As a further  example, we randomly generate $10^4$  unitary matrices $T$ sampled
according to  a Haar (uniform)  distribution on $\U(3)$  \cite{ZycKus94, Mez07}.
We compute our bound in these cases for some chosen, equal entropic indices, and
we compare  our results with the  bounds of PRZ, MU  and Deutsch in  the case of
R\'enyi entropy  (Fig.~\ref{RandomRenyi:fig}), with  the bounds of  Rastegin and
PRZ  in the  case of  Tsallis entropy  (Fig.~\ref{RandomTsallis:fig}),  and with
$\B^{\overline{CP}}$ in both cases (Fig.~\ref{RandomColes:fig}).

In Fig.~\ref{RandomRenyi:fig}  we plot the  bounds $\B_{\alpha,\alpha;\log}(c)$,
$\B^{MU}(c)$,  $\B_{\alpha;\log}^{PRZ}(T)$,   and  $\B^D(c)$  for   the  R\'enyi
entropic formulation of the~UP, in terms of the overlap $c\geq\frac{1}{\sqrt3}$,
when $\alpha=0.2, 0.8$ and $1.4$. We observe that:
\begin{itemize}
\item For  $\alpha=0.2$, our  bound improves MU  one in  the whole range  of the
  overlap. We find transformation matrices such that our bound improves PRZ one,
  although with a low frequency of occurrence.
\item For $\alpha=0.8$, our bound improves MU one when $c \geq \frac{1}{\sqrt2}$
  (Corollary~\ref{CorollaryQubit:cor}).   We find  transformation  matrices such
  that our bound improves PRZ one, with a frequency higher than for $\alpha=0.2$
  and increasing with $c$ as well.
\item For $\alpha=1.4$, our bound improves Deutsch one in the whole range of the
  overlap  (Corollary~\ref{ImproveDeutsch:cor}). Again,  we  find transformation
  matrices such  that our bound improves  PRZ one, with a  frequency higher than
  for $\alpha=0.8$ and increasing with $c$ as well.
\end{itemize}

\begin{figure}[htbp]
\centerline{\includegraphics[height=3.75cm]{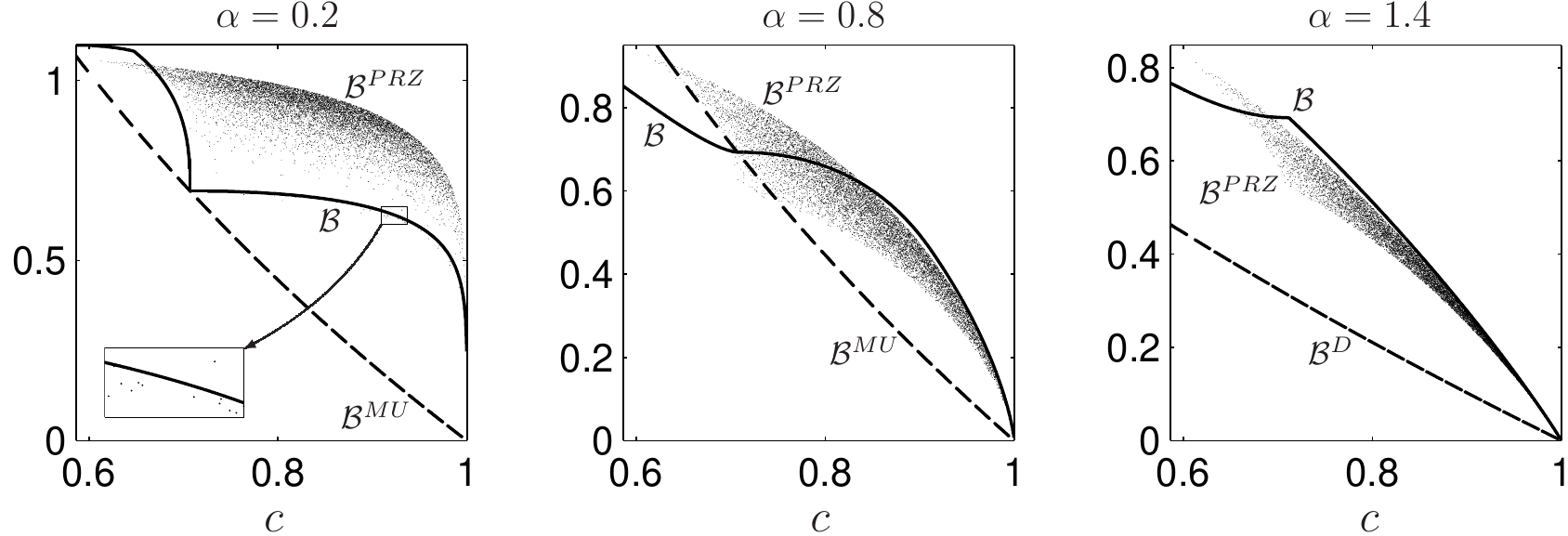}}
\caption{R\'enyi entropy case: bounds $\B_{\alpha,\alpha;\log}(c)$ (solid line),
  $\B^{MU}(c)$     (dashed     line,     left     and    middle     plots)     ,
  $\B_{\alpha;\log}^{PRZ}(T)$ (dots),  and $\B^D(c)$ (dashed  line, right plot),
  in terms of the overlap $c$ for $\alpha=0.2, 0.8$ and $1.4$.}
\label{RandomRenyi:fig}
\end{figure}

In Fig.~\ref{RandomTsallis:fig} we plot the bounds $\B_{\alpha,\alpha;\log}(c)$,
$\B_{\alpha,\alpha;\id-1}^R(c)$, and $\B_{\alpha;\log}^{PRZ}(T)$ for the Tsallis
entropic formulation of the~UP, in terms of the overlap $c\geq\frac{1}{\sqrt3}$,
when $\alpha=1,1.5$ and $2$. We observe that:
\begin{itemize}
\item  For   $\alpha=1$,  our   bound  improves  Rastegin   one  when   $c  \geq
  \frac{1}{\sqrt2}$      (Corollary~\ref{CorollaryQubit:cor}).       We     find
  transformation matrices such that our  bound improves PRZ one, with relatively
  high frequency of occurrence.
\item  For $\alpha=1.5$,  we find  transformation matrices  such that  our bound
  improves PRZ  one in a wider range  for the overlap and  with higher frequency
  than for $\alpha = 1$.
\item  For $\alpha=2$,  for all  the  sampled matrices  we find  that our  bound
  improves PRZ one in the whole range of the overlap.
\end{itemize}

\begin{figure}[htbp]
\centerline{\includegraphics[height=3.75cm]{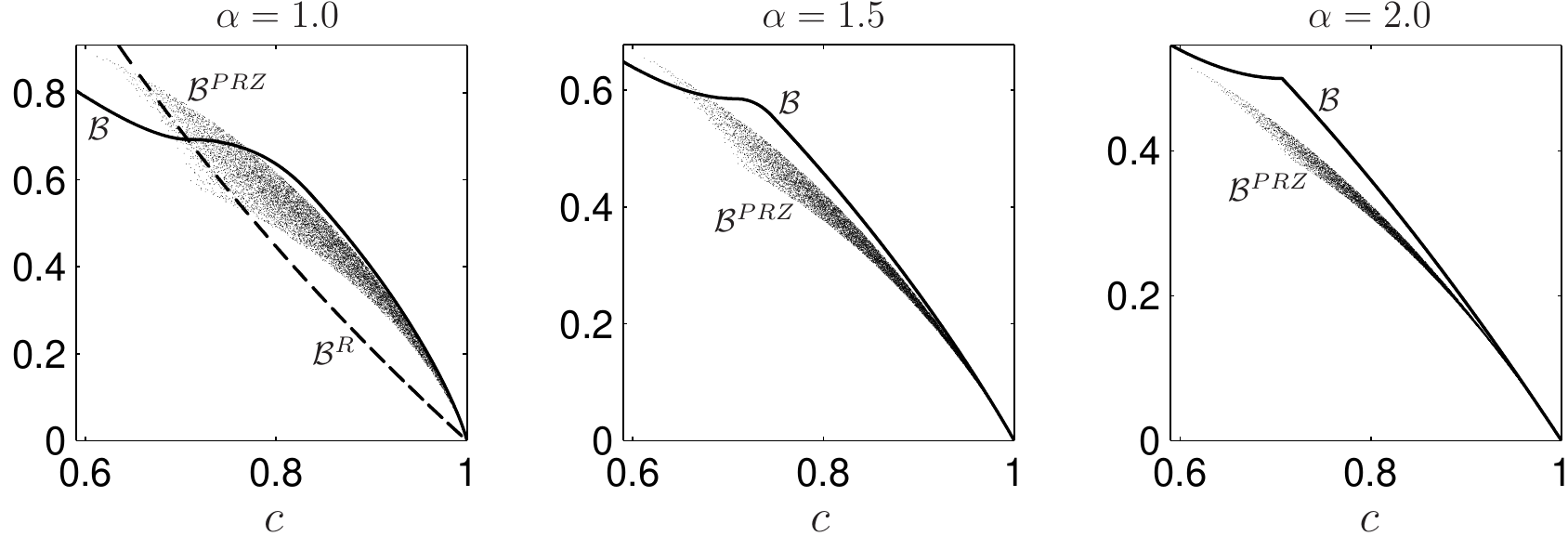}}
\caption{Tsallis  entropy   case:  bounds  $\B_{\alpha,\alpha;\id-1}(c)$  (solid
  line),   $\B_{\alpha,\alpha;\id-1}^R(c)$  (dashed   line,   left  plot),   and
  $\B_{\alpha;\log}^{PRZ}(T)$  (dots),   in  terms   of  the  overlap   $c$  for
  $\alpha=1,1.5$ and $2$.}
\label{RandomTsallis:fig}
\end{figure}

In  Fig.~\ref{RandomColes:fig}  we  plot the  bounds  $\B_{\alpha,\alpha;f}(c)$,
$\B^{MU}(c)$ or $\B_{\alpha,\alpha;\id-1}^R(c)$, and $\B^{\overline{CP}}(T)$ for
both R\'enyi and Tsallis entropic formulation of the~UP, in terms of the overlap
$c\geq\frac{1}{\sqrt3}$, when $\alpha=0.5$ and $1$.  We observe that:
\begin{itemize}
\item For any $\alpha$, our  bound improves $\B^{\overline{CP}}$ in a wide range
  of the overlap $c$.
\item In  the Tsallis  context, for  $\alpha \le \frac12$,  for all  the sampled
  matrices, we find an improvement of $\B^{\overline{CP}}$ in the whole range of
  the  overlap.  We  observe  that  the range  of  values of  $c$  for which  an
  improvement of the CP bound occurs, decreases with $\alpha$.
\end{itemize}

\begin{figure}[htbp]
\centerline{\includegraphics[height=3.75cm]{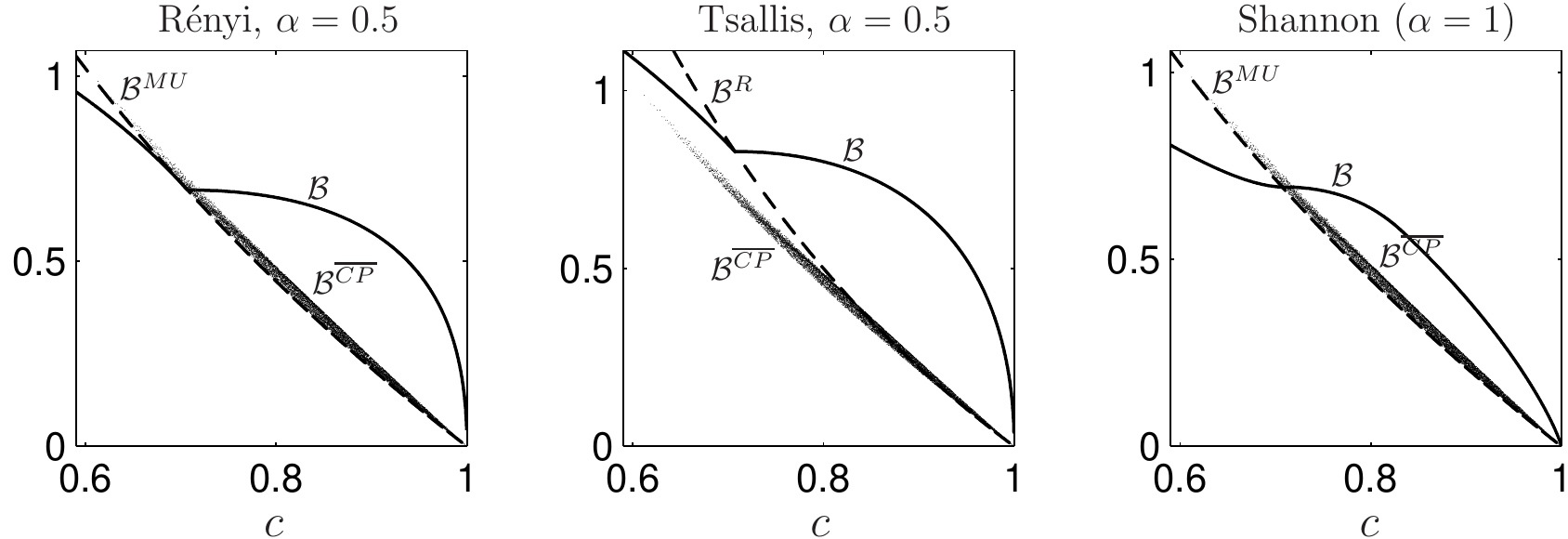}}
\caption{R\'enyi,     Tsallis    and     Shannon    entropy     cases:    bounds
  $\B_{\alpha,\alpha;f}(c)$       (solid       line),      $\B^{MU}(c)$       or
  $\B^{R}_{\alpha,\alpha;\id-1}$ (dashed line) , $\B^{\overline{CP}}(T)$ (dots),
  in terms  of the overlap  $c$ for $\alpha=0.5$  (R\'enyi and Tsallis)  and $1$
  (Shannon).}
\label{RandomColes:fig}
\end{figure}

We  notice that, as  MU, Deutsch,  Rastegin and  our bounds  depend only  on the
overlap $c$, then the same relative behaviors remain valid for dimensions higher
than~3 (at least for $c\geq\frac{1}{\sqrt3}$).  In contrast, that may not be the
case for the relation between CP, PRZ and our bound, since the formers depend on
the  whole transformation  matrix~$T$;  indeed,  we expect  an  increase of  the
predominance of PRZ and CP  over other $c$-dependent bounds.  However, our bound
is easier  to calculate than PRZ  one for instance  whose computation complexity
increases combinatorially with the dimension of the matrix~$T$.


\section{Concluding remarks}
\label{Discussions:sec}

In  this contribution  we  provide  a general  entropy-like  formulation of  the
uncertainty principle, for any pair of POVM  in the case of pure or mixed states
in       finite       dimensions.        The      sum       of       generalized
$(h,\phi)$-entropies~\eref{SalicruEnt:eq} associated to two POVMs is proposed as
measure of joint uncertainty, and lower bounds for that sum are searched for, in
terms of  the overlaps $\vec{c}$ between  the POVM, which in  a sense quantifies
the degree of incompatibility of the observables.  Our main result is summarized
in the  Proposition of Sec.~\ref{Main:sec}, where we  give a $\vec{c}$-dependent
lower bound for  the entropy-sum, leading to the  family of entropic uncertainty
relations~\eref{FEURs:eq}.  To  obtain this, we  follow the same approach  as de
Vicente and  S\'anchez-Ruiz appealing to  the Landau--Pollak inequality,  and we
solve the concomitant constrained  minimization problem, mainly in a geometrical
manner.  In this way, the  calculation of a $\vec{c}$-dependent bound reduces to
the  resolution  of  the  straightforward one-dimensional  minimization  problem
in~\eref{Cab:eq}.

Our uncertainty relation~\eref{FEURs:eq} generalizes previous similar results in
several ways, namely, it is valid for:
\begin{itemize}
\item  Salicr\'u generalized  entropic forms  [including R\'enyi~\eref{Renyi:eq}
  and  Tsallis~\eref{Tsallis:eq} entropies,  which are  obtained for  $\phi(x) =
  x^\lambda$   with   $h(x)   =    \frac{\log   x}{1-\lambda}$   and   $h(x)   =
  \frac{x-1}{1-\lambda}$, respectively],
\item  any  choice for  the  pair  of  entropic functionals  $(h_A,\phi_A)$  and
  $(h_B,\phi_B)$  (overcoming the  limitation due  to the  Riesz--Thorin theorem
  that  involves  conjugated pairs  of  indices  when  dealing with  the  family
  $F_\lambda$~\eref{Flambda:eq}  with the  same  $f$, which  is  mainly used  in
  related literature),
\item any  pair of positive operator valued measures, and
\item  both pure  and mixed  states  (which is  proved without  recourse to  the
  concavity property,  that, for instance,  R\'enyi entropy does not  fulfill in
  general).
\end{itemize}

Besides  we show that,  in the  case of  nondegenerate quantum  observables with
overlap   $c$,   the   bound   reduces  to   the   unidimensional   minimization
problem~\eref{Cabnondeg:eq}.  Moreover,  for values of the  overlap greater than
$\frac{1}{\sqrt2}$, our bound is $c$-optimal and it reduces to that of the qubit
($N=2$) case  (Corollary~\ref{CorollaryQubit:cor}).  In other  words, we improve
all $c$-dependent bounds in that range of the overlap.

In addition, we go further in the case of R\'enyi entropies and we find that our
bound  improves  Deutsch  one in  the  whole  range  of  values of  the  overlap
(Corollary~\ref{ImproveDeutsch:cor}), and  also that our bound  does not improve
Maassen--Uffink one for  values of the overlap lower than  or equal to $\frac12$
(Corollary~\ref{NoImprovedMU:cor}).    The   former   result   is   particularly
interesting  for entropic indices  above the  conjugacy curve  where, up  to our
knowledge,  Deutsch bound is  the only  known one  with an  analytic expression;
whereas  the  latter result  establishes  that  restricting  the domain  by  the
Landau--Pollak  inequality, leads to  a result  weaker than  using Riesz--Thorin
theorem.

Finally, in Sec.~\ref{Comparisons:sec}, we provide several examples that exhibit
an improvement with respect to known  results in the literature, in the cases of
R\'enyi and Tsallis entropies.

The extension of  our approach to take into account quantum  memory and for more
than two POVMs is currently under investigation.


\section*{Acknowledgments}
SZ  and MP  are very  grateful to  the R\'egion  Rh\^one-Alpes (France)  for the
grants that  enabled this work.  MP  and GMB also  acknowledge financial support
from CONICET (Argentina),  and warm hospitality during their  stay at GIPSA-Lab.
The authors thank  Pr. J.-F.  Bercher for useful discussions  about the class of
Salicr\'u  entropies.  The  authors acknowledge  anonymous referees  for helpful
comments.


\appendix


\section{Proof of the Proposition}
\label{Proposition:app}

Our aim is, given the probability vectors $p(A,\rho)$ and $p(B,\rho)$ associated
with   the  POVM   $A$   and  $B$   respectively,   to  minimize   the  sum   of
$(h,\phi)$-entropies subject to the Landau--Pollak inequality.  In this way, our
method    follows    and    advances    on    that    of    de    Vicente    and
S\'anchez-Ruiz~\cite{VicSan08}, and consists of two steps:
\begin{enumerate}
\item Minimization of $H_{(h,\phi)}$ subject to $\max_k p_k = P$.  At this step,
  the  two sets  of probabilities  are  treated separately.   Thus, denoting  by
  $H_{(h,\phi)}^{\min}(P)$ this  minimal entropy, we arrive at  the inequality \
  $H_{(h_A,\phi_A)}  \big( p(A,\rho)  \big) +  H_{(h_B,\phi_B)}  \big( p(B,\rho)
  \big)            \ge           H_{(h_A,\phi_A)}^{\min}(P_{A,\rho})           +
  H_{(h_B,\phi_B)}^{\min}(P_{B,\rho})$ where the right-hand side depends only on
  the two maximal probabilities.
\item     Minimization    of     \     $H_{(h_A,\phi_A)}^{\min}(P_{A,\rho})    +
  H_{(h_B,\phi_B)}^{\min}(P_{B,\rho})$   \   subject   to   the   Landau--Pollak
  inequality.
\end{enumerate}


\subsection{First  step: minimization  of  the $(h,\phi)$-entropy  subject to  a
  given maximum probability}

This  problem involves  looking for  the vectors  $p =  [p_1 \quad  \ldots \quad
p_N]^t  \in \P_N$  (the  set  of probability  vectors  in $\Rset_+^{\,N}$)  that
minimize  a  given $(h,\phi)$-entropy  under  the  constraint  that the  maximum
probability is provided\footnote{ In the context of Shannon entropy, the problem
  was  already solved  in Ref.~\cite{FedMer94},  using  the Karush--Khun--Tucker
  sufficient  conditions   for  convex  optimization   problems  \cite{CamMar09,
    Mil00}.}, i.e., we search for
\begin{equation}
\min_{p \in \P_N} H_{(h,\phi)}(p) = \min_{p \in \P_N} h\!\left( \sum_{k=1}^N
\phi(p_k) \right) \qquad \mbox{s.t. } \quad \max_k p_k = P
\label{Problem_MinEnt_MaxP:eq}
\end{equation}
Notice   that,   due   to   the  normalization   constraint,   one   necessarily
has\footnote{If $P  < \frac{1}{N}  \Rightarrow \sum_k  p_k \le N  P <  1$, which
  would contradict normalization.}
\begin{equation}
P \in \left[ \frac{1}{N} \; 1 \right]
\end{equation}
Note also that in the case $P  = \frac{1}{N}$, then all the $p_k$'s are equal to
$\frac{1}{N}$ (uniform distribution) and thus the problem becomes trivial.

Using  the  fact  that  the  function  to be  optimized  is  invariant  under  a
permutation of the  probability components, we can reduce  the dimensionality of
the problem in the following way: let  us fix $p_1\equiv P$ and define $q = [q_1
\quad \ldots \quad  q_{N-1}]^t \equiv [p_2 \quad \ldots  \quad p_N]^t$; then, to
solve  the optimization  problem~\eref{Problem_MinEnt_MaxP:eq} is  equivalent to
search for
\begin{equation}\left\{\begin{array}{ll}
\displaystyle \min_{q\in\PT_{\!P}} \, \varphi(q) & \mbox{if $\phi$ is concave}
\\[5mm]
\displaystyle \max_{q\in\PT_{\!P}} \, \varphi(q) & \mbox{if $\phi$ is convex}
\end{array}\right.
\label{Problem_MaxNorm_MaxP:eq}
\end{equation}
where we define
\begin{equation}
\varphi(q) =  \sum_{k=1}^{N-1} \phi(q_k)
\end{equation}
and we denote by $\PT_{\!P}$ the allowed domain for $q$, i.e.,
\begin{equation}
\hspace{-10mm} \PT_{\!P} = \left\{ q \in \Rset^{N-1} : \ 0 \leq q_k \leq P \ \wedge \
\sum_{k=1}^{N-1} q_k = 1-P \right\} = \HC_P \cap \HP_P
\label{polytope}
\end{equation}
with  \  $\HC_P =  [0  \; P]^{N-1}$  \  denoting  an $(N-1)$-dimensional  closed
hypercube,  and   $\displaystyle  \HP_P  =   \left\{  q  \in  \Rset^{N-1}   :  \
  \sum_{k=1}^{N-1} q_k  = 1-P \right\}$ corresponding  to an $(N-2)$-dimensional
hyperplane perpendicular to the vector $\vec{1}  = \left[ 1 \quad \cdots \quad 1
\right]^t$.  Notice that the point  $\frac{1-P}{N-1} \left( 1 \quad \cdots \quad
  1 \right)$ is both inside the hypercube $\HC_P$ and on the hyperplane $\HP_P$,
which guarantees that the intersection of those sets is not empty.

It  can be  seen  that $\PT_{\!P}$  is  a convex  polytope  embedded in  $\HC_P$
\cite{LarFlo09};  in other  words,  it is  a  convex body,  convex  hull of  its
vertices that are  the pure points of this convex (i.e.,  the points that cannot
be  written as convex  combination of  several points  of the  set) \cite{Egg58,
  Ber87:II}.

Next,  since  $\varphi$  is  a  strictly concave  (resp.\  convex)  function  on
$\Rset_+^{\,N-1}$,  it   is  also  concave  (resp.\  convex)   on  the  polytope
$\PT_{\!P}$.  It turns out that  $\varphi$ achieves its minimum (resp.\ maximum)
only on  one or several  of the extreme  points (or pure points)  of $\PT_{\!P}$
\cite{CamMar09, RobVar73}.  The problem consists  then in determining the set of
pure points of~\eref{polytope}.  Before studying  the case of arbitrary $N$, let
us illustrate what happens in the cases $N =  3$ and $N = 4$ \ (the case $N = 2$
is  trivial since $\PT_{\!P}$  reduces to  the point  $1-P$, and  the maximizing
probability vector is $(P,1-P)$ where $P$ should be between $\frac 12$ and 1).


\subsubsection{Case $N = 3$.}

Two different situations arise for the  intersection of the line $q_1+q_2 = 1-P$
with the square $0 \le q_1 \le P$, $0 \le q_2 \le P$:
\begin{itemize}
\item For  $\frac12 < P \le  1$, the line  intersects the square in  its ``lower
  corner'' or, in other words, the restriction of the line to the first quadrant
  is entirely inside the  square: $\HP_P\subset\HC_P$, then $\PT_{\!P}=\HP_P$ is
  the   whole  segment  between   the  points   $(1-P,0)$  and   $(0,1-P)$  [see
  Fig.~\ref{Dp_n3:fig} (left plot)]. These are the pure points, and both lead to
  the same extremal value for $\varphi$.
\item  For $\frac13 <  P \le  \frac12$, the  intersection of  the line  with the
  square reduces to  the segment linking the points  $(P,1-2 P)$ and $(1-2P,P)$,
  which are then the pure points of $\PT_{\!P}$ [see Fig.~\ref{Dp_n3:fig} (right
  plot)]. Both points lead to the same extremal value for $\varphi$.
\end{itemize}
Notice that the pure points are on the edges of the square.

\begin{figure}[htbp]
\begin{tabular}
{
>{}m{.51\textwidth}
>{}m{.44\textwidth}
}
\centerline{\includegraphics[width=.6\textwidth]{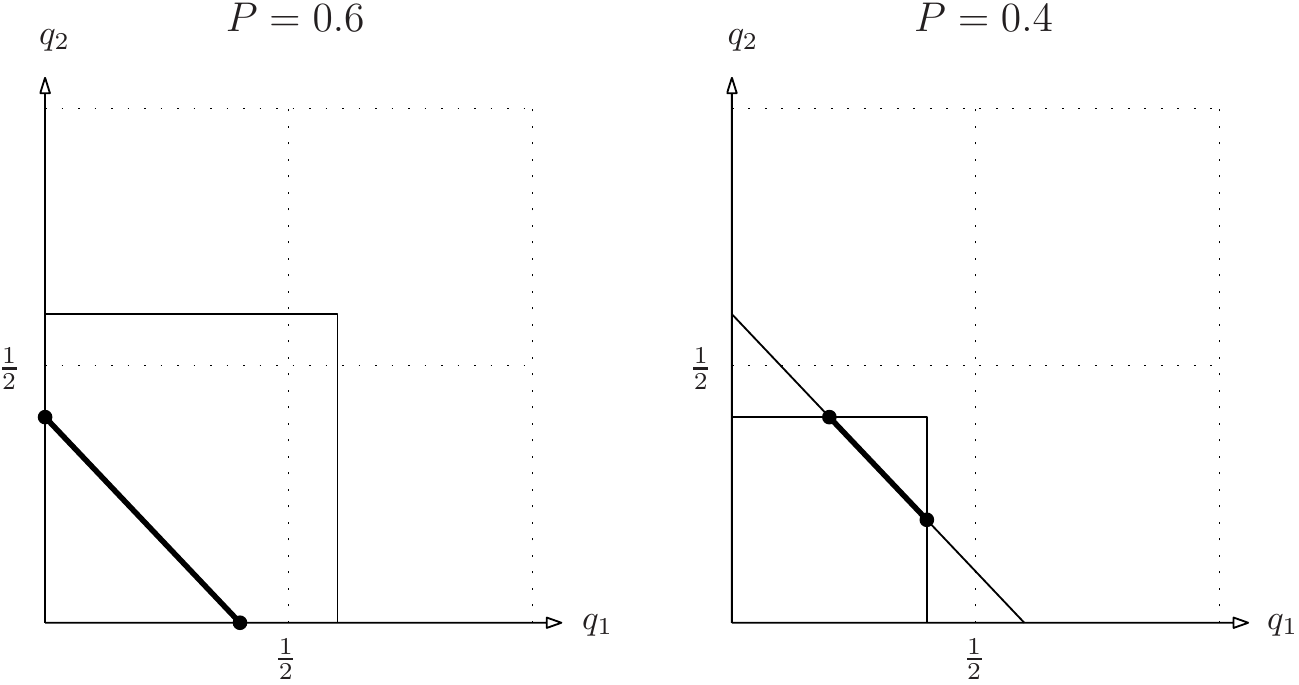}}
&
\caption{Domain $\PT_{\!P}$ (line in bold) in  the case $N=3$, for $P = 0.6$ and
  $0.4$ (from left to right). It is the intersection between the line $q_1+q_2 =
  1-P$ and the square $[0 \; P]^2$.  The pure points of $\PT_{\!P}$ are given by
  the dots.}
\label{Dp_n3:fig}
\end{tabular}
\end{figure}


\subsubsection{Case $N = 4$.}

Now,  three  different  situations  arise  for the  intersection  of  the  plane
$q_1+q_2+q_3 = 1-P$ with the cube $[0 \; P]^3$:
\begin{itemize}
\item  For $\frac12  <  P \le  1$, the  domain  $\PT_{\!P}$ is  the convex  body
  delimited by the triangle of vertices $(1-P,0,0)$, $(0,1-P,0)$ and $(0,0,1-P)$
  (triangle  and its interior);  the plane  intersects the  cube in  its ``lower
  corner'' or, in other words, the  restriction of the plane to the first octant
  is entirely inside the cube  [see Fig.~\ref{Dp_n4:fig} (left plot)].  The pure
  points are then  all the permutations of $(1-P,0,0)$, leading  all to the same
  extremal value for $\varphi$.
\item For $\frac13 < P \le \frac12$,  the plane intersects the six facets of the
  cube,  so that  $\PT_{\!P}$ is  the convex  body delimited  by the  hexagon of
  vertices $(P,1-2P,0)$, $(1-2P,P,0)$, $(0,1-2P,P)$, $(0,P,1-2P)$, $(1-2P,0,P)$,
  $(P,0,1-2  P)$, which are  the pure  points [see  Fig.~\ref{Dp_n4:fig} (middle
  plot)]. All of them lead to the same value for $\varphi$.
\item  For $\frac14  < P  \le \frac13$,  the plane  intersects the  cube  at its
  ``higher corner'',  so that  $\PT_{\!P}$ is the  convex body delimited  by the
  triangle of  vertices $(P,P,1-3  P)$, $(1-3 P,P,P)$  and $(P,1-3  P,P)$, these
  points being its pure  points [see Fig.~\ref{Dp_n4:fig} (right plot)].  Again,
  these points lead to the same value for $\varphi$.
\end{itemize}
Notice that the pure points are on the edges of the cube.

\begin{figure}[htbp]
\centerline{\includegraphics[width=\textwidth]{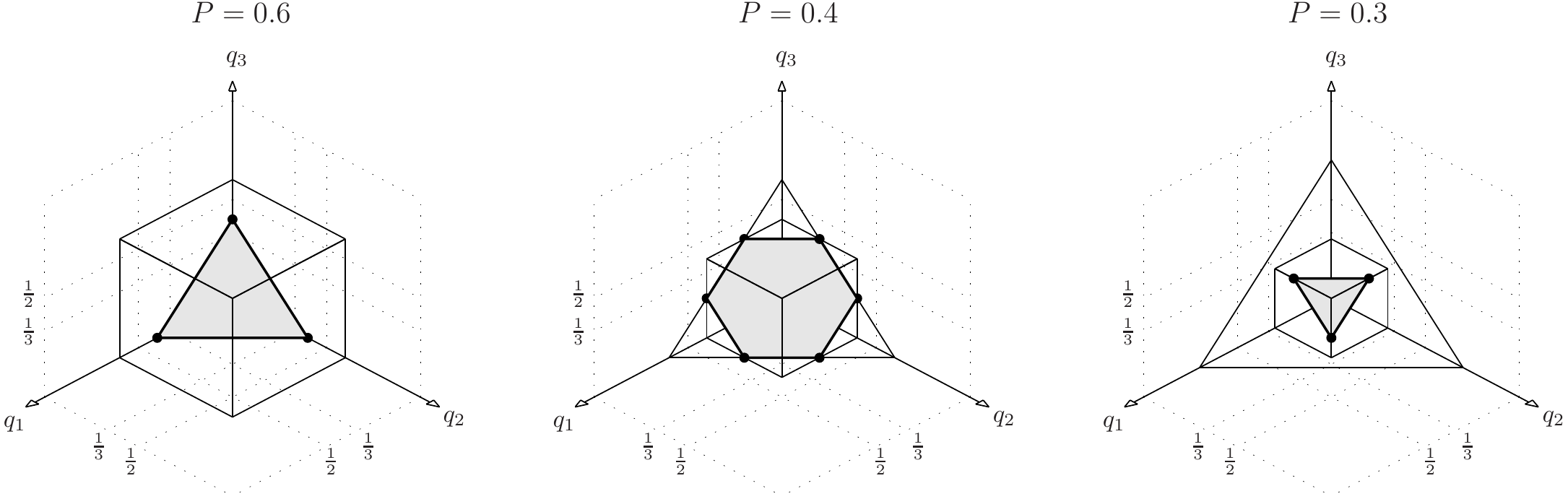}}
\caption{Domain $\PT_{\!P}$ (surface  in gray) in the case $N=4$,  for $P = 0.6,
  0.4$ and $0.3$ (from left to  right). It is the intersection between the plane
  $q_1+q_2+q_3 =  1-P$ and  the cube $[0  \; P]^3$.  The border of  the polytope
  $\PT_{\!P}$ is represented by bold lines, and its pure points are given by the
  dots.}
\label{Dp_n4:fig}
\end{figure}


\subsubsection{Arbitrary  $N=$  case: convex  polytope  $\PT_{\!P}$ and  minimal
  $(h,\phi)$-entropy.}


\paragraph{Pure points of the polytope $\PT_{\!P}$: }

As previously mentioned,  the intersection between a hypercube  and a hyperplane
is a  polytope, convex  hull of  its vertices that  are the  pure points  of the
polytope; moreover, the  vertices of the polytope are on  edges of the hypercube
\cite{LarFlo09}. Finding the vertices (i.e., the pure points) of such a polytope
is not an easy  task in general since the number of  vertices grows rapidly with
dimension $N$ \cite{LarFlo09}.  However,  the problem simplifies drastically due
to the regularity of the hypercube  $\HC_P = [0 \; P]^{N-1}$. Indeed, the $(N-1)
\, 2^{N-2}$ edges $E_P$ are of the form
\begin{equation}
E_P^{\downarrow} = \left\{ ( \underbrace{P,\ldots,P}_{M-1 \:\:
\mbox{\footnotesize times}} , s , \underbrace{0,\ \ldots \ ,0}_{N-M-1 \:\:
\mbox{\footnotesize times}} ) \ , \quad 0\leq s<P \right\}
\label{EP:eq}
\end{equation}
for  every $M  =  1, 2,  \ldots,  N-1$, where  $\cdot^{\downarrow}$ denotes  the
rearrangement of the $(N-1)$-uplet (components put in decreasing order).

A point in $E_P$  is a vertex of the polytope $\PT_{\!P}$  if it also belongs to
$\HP_P$,  that is for  $s^*\in[0 \;  P)$ such  that $(M-1)  P +  s^* =  1-P$, or
$M=\frac{1-s^*}{P}$, which is greater that  $\frac 1P-1$, and less than or equal
to $\frac 1P$.  Since $M$ is an integer  we finally find that, given  a value of
$P$, the pure points are such that
\begin{equation}
M = \left\lfloor \frac{1}{P} \right\rfloor \qquad \mbox{and} \qquad s^* = 1 -
\left\lfloor \frac{1}{P} \right\rfloor P
\label{mEP:eq}
\end{equation}
where $\lfloor \cdot \rfloor$ denotes the floor part.

This allows us to conclude that the  edges of $\HC_P$ contain at most one vertex
of $\PT_{\!P}$ (which is intuitive since  no facet of $\HC_P$ is parallel to the
hyperplane $\HP_P$) and that $\PT_{\!P}$ is the  convex hull of the set of the \
$(N-M) \left(  \!\! \begin{array}{c}  N-1\\[.5mm] M-1 \end{array}  \!\! \right)$
points  that  belong  to $E_P$,  Eq.~\eref{EP:eq},  for  $s$  and $M$  given  in
Eq.~\eref{mEP:eq}. This  has been  illustrated in the  particular cases $N  = 3$
(with $M=1$  and $2$  from left to  right in  Fig.~\ref{Dp_n3:fig}) and $N  = 4$
(with $M=1,2$ and $3$ from left to right in Fig.~\ref{Dp_n4:fig}).


\paragraph{Optimal  vector   and  minimal  entropy: }

As  previously recalled, $\varphi$  being strictly  concave (resp.\  convex), it
achieves its minimum (resp.\ maximum)  on the polytope (convex body) $\PT_{\!P}$
only  in  its vertices  (pure  points).  In  other  words,  the minimal  entropy
solution for the original problem~\eref{Problem_MinEnt_MaxP:eq} is achieved only
for the probability vectors of the form
\begin{equation}
p^{\downarrow} = \big[ \, \underbrace{P \quad \ldots \quad P}_{\left\lfloor
\frac{1}{P} \right\rfloor \:\: \mbox{\footnotesize times}} \quad {\textstyle 1 -
\left\lfloor \frac{1}{P} \right\rfloor} P \quad \underbrace{0 \quad \ldots \quad
0}_{N-\left\lfloor \frac{1}{P} \right\rfloor-1 \:\: \mbox{\footnotesize times}}
\big]^t
\label{Vecp_MinEnt_MaxP:eq}
\end{equation}
and its expression $H_{(h,\phi)}^{\min}(P) \:\: \equiv \displaystyle \min_{p \in
  \P_N: \max_k p_k = P} \: H_{(h,\phi)}(p)$ is given by
\begin{equation}
H_{(h,\phi)}^{\min}(P) \: = \: h \! \left( \left\lfloor \frac{1}{P}
\right\rfloor \phi(P) + \phi \! \left( 1 - \left\lfloor \frac{1}{P}
\right\rfloor P \right) \right)
\label{MinEnt_MaxP:eq}
\end{equation}
where $P \in \left[ \frac{1}{N} \; 1\right]$.

We  can verify {\em  a posteriori}  the solution  obtained for  the minimization
problem, using the Schur-concavity  of $(h,\phi)$-entropies.  Indeed, vector $p$
defined by Eq.~\eref{Vecp_MinEnt_MaxP:eq}  majorizes all the probability vectors
with maximal  probability equal  to $P$,  and thus its  entropy is  minimal over
these probability vectors.


\subsection{Second step: minimization of the sum of minimal $(h,\phi)$-entropies
  subject to the Landau--Pollak inequality}

Recall   that  Landau--Pollak   inequality  links   the   maximal  probabilities
$P_{A,\rho}$  and   $P_{B,\rho}$  corresponding  to  the  POVMs   $A$  and  $B$,
respectively~\cite{BosZoz14}.  We now address the problem of minimization of the
sum of minimal  $(h,\phi)$-entropies, which is written in  terms of $P_{A,\rho}$
and $P_{B,\rho}$, under that inequality constraint.  We first analyze the domain
where the pair $(P_{A,\rho},P_{B,\rho})$ lives  and then the behavior of the sum
of minimal  entropies within this domain.   This allows us  to slightly simplify
the problem.


\subsubsection{Representation of the Landau--Pollak inequality domain.}

Following  our   previous  work~\cite{BosZoz14},  it   can  be  seen   that  the
Landau--Pollak   inequality  constrains  the   pair  of   maximal  probabilities
$\left(P_{A,\rho},P_{B,\rho} \right)$ in the domain:
\begin{equation}
\hspace{-25mm} \Dset_\lp(\vec{c}) \! = \! \left\{ \! (P_A,P_B) \! \in \! \left[
\frac{1}{N_A} \; c_A^{\ 2} \right] \! \times \! \left[ \frac{1}{N_B} \; c_\B^{\
2} \right] :\ P_B \le g_{c_{A,B}}\big(P_A\big) \: \mbox{when} \: P_A \ge
c_{A,B}^{\ 2} \right\}
\label{Dlp:eq}
\end{equation}
where $\vec{c} = (c_A,c_B,c_{A,B})$ and
\begin{equation}
g_c(x) = \cos^2\left(\arccos c - \arccos\sqrt{x} \right)
\end{equation}
If $c_B^{\  2} \le g_{c_{A,B}}(c_A^{\  2})$, the allowed domain  becomes $\left[
  \frac{1}{N_A} \;  c_A^{\ 2} \right]  \times \left[ \frac{1}{N_B} \;  c_B^{\ 2}
\right]$.  This is represented in Fig.~\ref{Dlp:fig}.

\begin{figure}[htbp]
\centerline{\includegraphics[width=.75\textwidth]{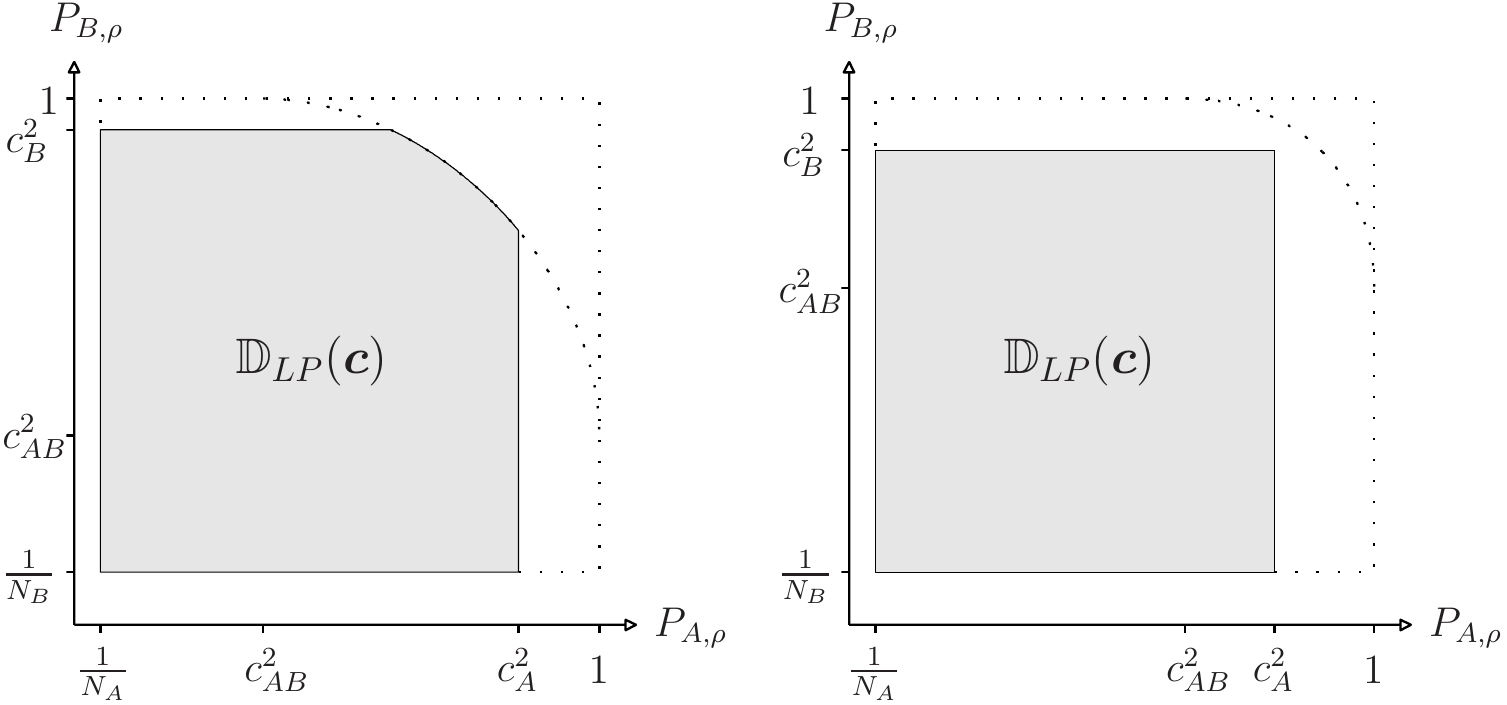}}
\caption{Representation  (shaded  region)  of  the  domain  $\Dset_\lp(\vec{c})$
  \eref{Dlp:eq}  for   pairs  of  maximal   probabilities  when  $c_B^{\   2}  >
  g_{c_{A,B}}(c_A^{\  2})$ (left)  and  $c_B^{\ 2}  \le g_{c_{A,B}}(c_A^{\  2})$
  (right).}
\label{Dlp:fig}
\end{figure}


\subsubsection{Minimal entropies sum.}

We have reduced the problem to solve
\begin{equation}
\min_{(P_{A,\rho},P_{B,\rho})\in \Dset_\lp(\vec{c})} \big\{
H_{(h_A,\phi_A)}^{\min}(P_{A,\rho}) + H_{(h_B,\phi_B)}^{\min}(P_{B,\rho}) \big\}
\end{equation}
for    given    $A$,    $B$,    $(h_A,\phi_A)$    and    $(h_B,\phi_B)$,    with
$H_{(h,\phi)}^{\min}(P)$ given by Eq.~\eref{MinEnt_MaxP:eq}.  For any $M = 1, 2,
\ldots,N-1$, and  for any $P_1$ and  $P_2$ such that $\frac{1}{M+1}  \le P_1 \le
P_2 \le  \frac{1}{M}$ we have $[P_1  \: \ldots \: P_1  \quad 1-M P_1  \quad 0 \:
\ldots \: 0]^t  \prec [P_2 \: \ldots \:  P_2 \quad 1-M P_2 \quad 0  \: \ldots \:
0]^t$   and  thus,   from   the  Schur-concavity   of  the   $(h,\phi)$-entropy,
$H_{(h,\phi)}^{\min}(P_1)   \ge  H_{(h,\phi)}^{\min}(P_2)$.   In   other  words,
function  $P \mapsto  H_{(h,\phi)}^{\min}(P)$  is decreasing  in each  intervals
$\left(  \frac{1}{M+1}  \; \frac{1}{M}  \right)$  and  thus,  by continuity,  in
$\left( 0  \; 1  \right]$.  This is  illustrated in Fig.~\ref{HaMin:fig}  in the
case of R\'enyi and Tsallis entropies.

\begin{figure}[htbp]
\centerline{\includegraphics[width=.975\textwidth]{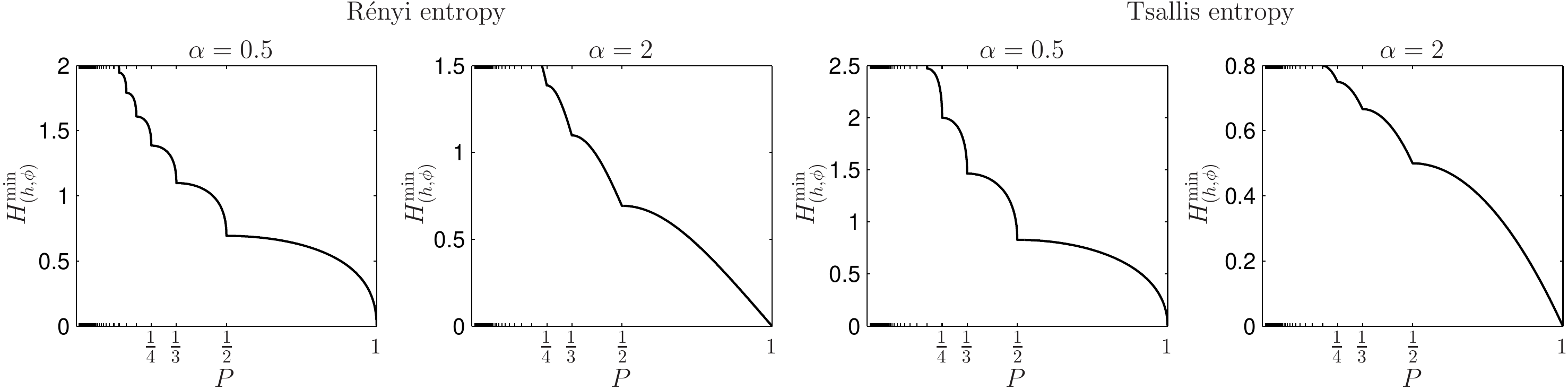}}
\caption{Decreasing behavior of  the function $H_{(h,\phi)}^{\min}(P)$ versus $P
  \in \left( 0 \;  1 \right]$, in the case of R\'enyi  entropy (first and second
  plots) and in the case of  Tsallis entropy (third and fourth plots).  Here the
  entropic index are $\alpha = 0.5$ or $2$, as indicated.}
\label{HaMin:fig}
\end{figure}

Reasoning  by  fixing  $P_{A,\rho}$   and  minimizing  the  entropies  sum  over
$P_{B,\rho}$ and reversing the roles of  $A$ and $B$, we immediately obtain that
the minimum is achieved when:
\begin{itemize}
\item $\left(P_{A,\rho},  P_{B,\rho}\right) = (c_A^{\ 2},c_B^{\  2})$ if $c_B^{\
    2} \le g_{c_{A,B}}(c_A^{\ 2})$. Thus, the minimum takes the analytical form
$$
H_{(h_A,\phi_A)}^{\min}(c_A^{\ 2}) + H_{(h_B,\phi_B)}^{\min}(c_B^{\ 2})
$$
or
\item $\left(P_{A,\rho}, P_{B,\rho}\right)$ is  in the curve $ \left(P_{A,\rho},
    g_{c_{A,B}}(P_{A,\rho})    \right)$     with    $P_{A,\rho}    \in    \left[
    g_{c_{A,B}}(c_B^{\   2}),    c_A^{\   2}   \right]$   if    $c_B^{\   2}   >
  g_{c_{A,B}}(c_A^{\ 2})$.
\end{itemize}
Let us define the angles
$$
\gamma_A \equiv \arccos c_A, \quad \gamma_B \equiv \arccos c_B, \quad
\mbox{and} \quad \gamma_{A,B} \equiv \arccos c_{A,B}
$$
the one-to-one mapping
$$
P_{A,\rho} \equiv \cos^2 \theta \qquad \mbox{with} \qquad \theta \in [\gamma_A
\;\gamma_{A,B} - \gamma_B],
$$
leading to
$$
g_{c_{A,B}}(P_{A,\rho}) = \cos^2(\gamma_{A,B}-\theta)
$$
with \  $\gamma_{A,B} -  \theta \in [\gamma_A  \;\gamma_{A,B} -  \gamma_B]$, and
function
\begin{equation}
\D_{(h,\phi)}(\theta) \equiv h \! \left( \left\lfloor
\frac{1}{\cos^2\theta}\right\rfloor \phi \! \left( \cos^2 \theta \right) + \phi
\! \left( 1 - \left\lfloor \frac{1}{\cos^2\theta}\right\rfloor \cos^2\theta
\right) \right),
\end{equation}
With these notations,
\begin{itemize}
\item Condition  $c_B^{\ 2} \le g_{c_{A,B}}(c_A^{\ 2})$  simplifies to $\gamma_B
  \ge \gamma_{A,B} - \gamma_A$,\vspace{4mm}
\item  $H_{(h_A,\phi_A)}^{\min}(c_A^{\ 2})  =  \D_{(h_A,\phi_A)}(\gamma_A)$ (and
  similarly for $B$),\vspace{4mm}
\item  $H_{(h_A,\phi_A)}^{\min}(P_{A,\rho})   =  \D_{(h_A,\phi_A)}(\theta)$  and
  $H_{(h_B,\phi_B)}^{\min}(g_{c_{A,B}}(P_{A,\rho}))                             =
  \D_{(h_B,\phi_B)}(\gamma_{A,B}-\theta)$
\end{itemize}
Thus, the minimal entropies sum is given by
\begin{equation*}\left\{\begin{array}{lll}
\displaystyle \D_{(h_A,\phi_A)}(\gamma_A) + \D_{(h_B,\phi_B)}(\gamma_B) &
\mbox{if} \quad \gamma_{A,B} \le \gamma_A + \gamma_B \\[5mm]
\displaystyle \min_{\theta \in [\gamma_A, \gamma_{A,B} - \gamma_B]} \left(
\D_{(h_A,\phi_A)}(\theta) + \D_{(h_B,\phi_B)}(\gamma_{A,B}-\theta) \right) &
\mbox{otherwise}
\end{array}\right.\end{equation*}
proving the Proposition.  Note that the cosine being increasing (in the interval
$\theta$  lies in),  the decreasing  property of  $H_{h,\phi}^{\min}(P)$ implies
that $\D_{(h,\phi)}(\theta)$ is increasing vs $\theta$.


\section{Proof of Corollary~\ref{CorollaryQubit:cor}}
\label{CorollaryQubit:app}

Remember that in this case, we have $N_A=N_B=N$, $c_A=c_B=1$ and $c_{A,B}=c$.

In Ref.~\cite{ZozBos13} we solved the problem in the case of the qubit ($N = 2$)
for pure states and for the R\'enyi entropy. It appears that:
\begin{itemize}
\item  This result  extends  for arbitrary  pairs  of $(h,\phi)$-entropies;  the
  approach~\cite[Appendix A]{ZozBos13}  extends step by step  to such entropies,
  where  the concavity (resp.\  convexity) of  $\phi_B$ is  used instead  of the
  convexity  of  the  mapping  $z  \mapsto  \frac{\sum_k  |z_k|^\beta}{\beta-1}$
  (see~\cite[Eq.~(A.13)]{ZozBos13})    and   where   the    Schur-concavity   of
  $H_{h_B,\phi_B}$ is used to  finish the proof (see~\cite[Eqs.~(A.14)-(A.19) \&
  App. A.3.2]{ZozBos13}), which allows  to consider functions $h_B$ and $\phi_B$
  nonnecessarily differentiable.
\item The  extended bound for the  qubit and pure states  writes precisely under
  the  form  Eq.~\eref{Cabnondeg:eq} where  $c  >  \frac{1}{\sqrt{2}}$ and  thus
  $\left\lfloor  \frac{1}{\cos^2  \theta} \right\rfloor  =  1$  (the  case $c  =
  \frac{1}{\sqrt{2}}$ is recovered by continuity).
\item The  minimizing pure states of Proposition~2  of~\cite{ZozBos13} expressed
  through the optimal angles $\theta$ hold, where these angles clearly depend of
  the pairs of functionals $(h_A,\phi_A)$ and $(h_B,\phi_B)$.
\item Due to the coincidence  of bound~\eref{Cabnondeg:eq} and the optimal bound
  for pure states, this bound remains  optimal in the mixed states (a pure state
  being a particular pure state).
\end{itemize}

Since the case $N=2$ is already treated, let us concentrate on $N \ge 3$.

In the context of pure states,  one has $\rho = |\Psi\rangle \langle\Psi|$ where
$|\Psi\rangle$ is  an element  of an $N$-dimensional  Hilbert space.   Using the
notation of  Ref.~\cite{ZozBos13}, the state  $|\Psi\rangle$ can be  expanded on
the eigenbases of $A$  and $B$ under the form \ $\left|  \Psi \right\rangle \: =
\: \sum_{i=1}^N \psi_i \, |a_i\rangle \: = \: \sum_{j=1}^N \widetilde{\psi}_j \,
|b_j\rangle$.   Thus   \  $p_i(A,\rho)  =  |\psi_i|^2$  and   \  $p_j(B,\rho)  =
|\widetilde{\psi}_j|^2$.  Moreover, arranging  the complex coefficients $\psi_i$
and $\widetilde{\psi}_j$ in column vectors,  $\psi = \left[ \psi_1 \cdots \psi_N
\right]^t$   and  \   $\widetilde{\psi}  =   \left[   \widetilde{\psi}_1  \cdots
  \widetilde{\psi}_N \right]^t$, one  can see that these vectors  are linked via
$\widetilde{\psi}  =  T \psi$  where  $T$  is  the transformation  matrix  whose
elements are defined in Eq.~\eref{Relation_psi_psitilde:eq}.

Now,  let  us consider  $N  \times  N$  unitary matrices  of  the  form \  $T  =
\left[  \begin{array}{cc} T^{(2)}  & 0  \\ 0  & T^{(N-2)}  \end{array} \right]$,
where $T^{(n)}$  stands for an  $n \times n$  unitary matrix, and we  impose the
largest-modulus element  of $T$  to be  ``located'' in $T^{(2)}$,  that is  $c =
\max_{i,j} |T_{ij}|  = \max_{i,j} |T_{ij}^{(2)}|$.   This last condition  can be
fulfilled  only  if  $N  \geq  4$  because  one  must  have  $c  \ge  \max_{i,j}
|T_{ij}^{(N-2)}|   \in   \left[  \frac{1}{\sqrt{N-2}}   \;   1  \right]$.    Let
$|\Psi^{(2)}\rangle=\psi^{(2)}_1 \,  |a_1\rangle+\psi^{(2)}_2 \, |a_2\rangle$ be
a  minimizing  qubit  pure  state  corresponding to  the  transformation  matrix
$T^{(2)}$   so   that    $H_{(h_A,\phi_A)}   \big(   p(A,\rho^{(2)})   \big)   +
H_{(h_B,\phi_B)}          \big(         p(B,\rho^{(2)})          \big)         =
\widetilde{\B}_{(h_A,\phi_A),(h_B,\phi_B);2}(c)$      with     $\rho^{(2)}     =
|\Psi^{(2)}\rangle \langle\Psi^{(2)}|$.   Consider the density  operator $\rho =
|\Psi\rangle \langle\Psi|$ build up  from the extended pure state $|\Psi\rangle$
such  that  its  vector  of  coefficients is  $\psi  =  \left[  \begin{array}{c}
    \psi^{(2)} \\ 0 \end{array}  \right]$.  Then one has $H_{(h_A,\phi_A)} \big(
p(A,\rho)  \big) +  H_{(h_B,\phi_B)}  \big( p(B,\rho)  \big) =  H_{(h_A,\phi_A)}
\big(  p(A,\rho^{(2)}) \big)  + H_{(h_B,\phi_B)}  \big( p(B,\rho^{(2)})  \big) =
\widetilde{\B}_{(h_A,\phi_A),(h_B,\phi_B);2}(c) = \B_{(h_A,\phi_A),(h_B,\phi_B)}
(c)$.   The last  equality comes  from the  coincidence between  the $c$-optimal
bound  for  the  qubit  case (see  above),  and  expression~\eref{Cabnondeg:eq}.
Finally,   by    definition   of   $c$-optimal    bound,   one   has    both   $
\B_{(h_A,\phi_A),(h_B,\phi_B)}(c)                                             \le
\widetilde{\B}_{(h_A,\phi_A),(h_B,\phi_B);N}(c)$              and              $
\B_{(h_A,\phi_A),(h_B,\phi_B)}(c)  = H_{(h_A,\phi_A)}  \big(  p(A,\rho) \big)  +
H_{(h_B,\phi_B)}           \big(          p(B,\rho)           \big)          \ge
\widetilde{\B}_{(h_A,\phi_A),(h_B,\phi_B);N}(c)$,   proving  the  $c$-optimality
of~\eref{Cabnondeg:eq} when $c > \frac{1}{\sqrt{2}}$ and $N \geq 4$.

The problem  of the $c$-optimality  of the  bound for $N  = 3$ remains  open. We
suspect that it is so but we have not been able to prove this yet.


\section{Proof of Corollary~\ref{ImproveDeutsch:cor}}
\label{CorollaryDeutsch:app}

It can be seen that our bound~\eref{Cabnondeg:eq} in the case of R\'enyi entropy
when    $\alpha$    and    $\beta$    are   sufficiently    large,    gives    $
\B_{\infty,\infty;\log}(c)   =   \min_{\theta   \in   [0  \;   \gamma]}   [   -2
\log(\cos\theta)  - 2  \log\left(\cos(\gamma-\theta)\right)]$.   The minimum  is
attained  for $\theta  = \frac{\gamma}{2}$  so  that we  recover Deutsch  bound:
$\B_{\infty,\infty;\log}(c) = -2  \log \left( \frac{1+c}{2}\right) = \B^{D}(c)$.
Now, consider our bound $\B_{\alpha,\beta;\log}(c)$ which is the solution of the
minimization~\eref{Cabnondeg:eq},  and  the  probability  $P_A$  for  which  the
minimum is attained.  Since R\'enyi entropy decreases versus the entropic index,
we  have  $ \B_{\alpha,\beta;\log}(c)  =  R_\alpha^{\min}(P_A) +  R_\beta^{\min}
\left(g_c(P_A)\right)           \ge           R_{\infty}^{\min}(P_A)           +
R_{\infty}^{\min}\left(g_c(P_A)\right) \ge \B_{\infty,\infty;\log}(c) = \B^D(c)$
where $R_{\lambda}^{\min} \equiv  H_{\left( \frac{\log}{1-\lambda} , \id^\lambda
  \right)}^{\min}$, that proves that our bound improves Deutsch one.


\section{Proof of Corollary~\ref{NoImprovedMU:cor}}
\label{CorollaryMU:app}

Let us  consider the extreme  pair of indices  $(\alpha,\beta) = (0,0)$,  and go
back to expression~\eref{Cabnondeg:eq} for the bound,
$$
\B_{0,0;f} (c) = \min_{P_{A,\rho}  \in [c^2 \; 1]} \left[ R_0^{\min}(P_{A,\rho})
  + R_0^{\min}\left(g_c(P_{A,\rho})\right) \right]
$$
By    symmetry     of    the     quantity    in    square     brackets,    since
$g_c\left(\frac{1+c}{2}\right) = \frac{1+c}{2}$, one can restrict the search for
$P_{A,\rho}$ to the interval $\left[ c^2 \; \frac{1+c}{2} \right]$. Then:
\begin{itemize}
\item  For  $P_{A,\rho}   =  c^2$  one  has  $g_c(P_{A,\rho})   =  1$  and  thus
  $R_0^{\min}\left(g_c(P_{A,\rho})\right)  = 0$ while  $R_0^{\min}(P_{A,\rho}) =
  \log \left( \left\lceil \frac{1}{c^2} \right\rceil \right)$.
\item  For  $P_{A,\rho}  \in  \left(  c^2  \;  \frac{1+c}{2}  \right]$  one  has
  $g_c(P_{A,\rho}) \in \left[ \frac{1+c}{2}  \; 1 \right) \subset \left( \frac12
    \; 1 \right)$ and thus $R_0^{\min}\left(g_c(P_{A,\rho})\right) = \log 2$.  A
  rapid inspection of $R_0^{\min}(P_{A,\rho})$ allows  one to prove that in this
  interval it decreases  vs $P_{A,\rho}$ and that the minimum  is also $\log 2$.
  Thus,  $$\min_{P_{A,\rho}  \in \left(  c^2  \;  \frac{1+c}{2} \right]}  \left[
    R_0^{\min}(P_{A,\rho}) +  R_0^{\min}\left(g_c(P_{A,\rho})\right) \right] = 2
  \log 2$$
\end{itemize}
Therefore
$$
\B_{0,0;\log}(c) = \min  \left\{ 2 \log 2 \: , \:  \log \left( \left\lceil \frac{1}{c^2}
    \right\rceil \right) \right\} .
$$
Now,  when $c  \le \frac12$,  we have  $\B^{MU} (c)  = -2  \log c  \ge \log  4 =
\B_{0,0;\log}(c)$.  Moreover,  in this case $\B_{0,0;\log}(c) =  2 \, R_0^{\min}
\left(  \frac{1+c}{2} \right)$  so  that  by using  the  decreasing property  of
$R_{\lambda}^{\min}$ vs $\lambda$ we obtain
$$
\B_{\alpha,\beta;\log}(c)  \le  R_\alpha^{\min}\left(  \frac{1+c}{2}  \right)  +
R_\beta^{\min}\left(   \frac{1+c}{2}  \right)   \le  2   \,   R_0^{\min}  \left(
  \frac{1+c}{2} \right) = \B_{0,0;\log}(c) \le \B^{MU}(c)
$$
that concludes the proof.


\section*{References}
\bibliography{Uncertainty_qunit_R1}
\bibliographystyle{unsrt}

\end{document}